\documentclass[12pt,centertags,reqno]{amsart}

\usepackage{latexsym}
\usepackage{etex}
\usepackage[utf8]{inputenc} %per le lettere accentate
\usepackage[T1]{fontenc}
\usepackage[english]{babel}
\usepackage{graphicx} %permette di inserire grafici/diagrammi
\usepackage{mathrsfs}
\usepackage[foot]{amsaddr}
\usepackage{amsmath,amssymb,amsthm}
\usepackage{amsfonts}
\usepackage{amsthm}
\usepackage{epsfig}
\usepackage{empheq}
\usepackage{enumerate} %vari elenchi numerati
\usepackage{enumitem} %vari elenchi puntati
\usepackage{bigints} %pacchetto integrali
\usepackage{ifsym}
\usepackage{url}
\usepackage{mathrsfs, mathtools}
\usepackage{stmaryrd}
\usepackage[round]{natbib}
\usepackage{hyperref} %collegamento dalla citazione alla bibliografia ma anche dall'indice al capitolo
\usepackage{float}   
\usepackage{fancyhdr}
\usepackage{bbm}
\usepackage{marvosym}
\usepackage{caption}
\usepackage{textcomp}
\usepackage{pdfpages}
\usepackage{color,graphicx}
\usepackage{booktabs}
\usepackage{orcidlink}

\DeclareMathOperator*{\esssup}{ess\,sup}
\mathtoolsset{showonlyrefs}

\numberwithin{equation}{section} \makeatletter
\renewcommand{\subsection}{\@startsection
{subsection}{2}{0mm}{\baselineskip}{-0.25cm}
{\normalfont\normalsize\bf}} \makeatother

\newcommand{\N}{\mathbb{N}} %ridefinisce comando, scrivi \N per il simbolo dei naturali
\newcommand{\R}{\mathbb{R}} %ridefinisce comando, scrivi \R per il simbolo dei reali
\newcommand{\Rp}{\mathbb{R}^{+}}

\newcommand{\pd}[2]{\dfrac{\partial#1}{\partial#2}}
\newcommand{\pds}[2]{\dfrac{\partial{^{2}}#1}{\partial{#2}^{2}}}

\newtheorem{theorem}{Theorem}[section]
\newtheorem{lemma}[theorem]{Lemma}
\newtheorem{corollary}[theorem]{Corollary}
\newtheorem{definition}[theorem]{Definition}
\newtheorem{remark}[theorem]{Remark}
\newtheorem{proposition}[theorem]{Proposition}

\newtheorem{ass}[theorem]{Assumption}

\newtheorem{problem}[theorem]{Problem}

\def \I {{\mathbf 1}}
\def \P {\mathbf P}

\def \F {\mathcal F}
\def \bF {\mathbb F}

\newcommand{\ud}{\mathrm d}
\newcommand{\ds}{\displaystyle}
\newcommand{\esp}[2][\mathbb E] {#1\left[#2\right]}

\begin{document}

%\author{
%\name{Alessandra Cretarola\orcidlink{0000-0003-1324-9342}\textsuperscript{a}\thanks{CONTACT A. Cretarola. Email: alessandra.cretarola@unich.it} and Benedetta Salterini\orcidlink{0000-0001-6950-4766}\textsuperscript{b}}
%\affil{\textsuperscript{a}Department of Economic Studies,
% University “G. D’Annunzio” of Chieti-Pescara, Viale Pindaro, 42, I-65127 Pescara, Italy; \textsuperscript{b} Department of Mathematics and Computer Science,
% University of Firenze, Viale Morgagni, 67/A, I-50134 Firenze, Italy}
%}

\author[A.~Cretarola]{Alessandra Cretarola\, \orcidlink{0000-0003-1324-9342}}
\address{Alessandra Cretarola({\large \Letter}), Department of Economics,
 University “G. D’Annunzio” of Chieti-Pescara, Viale Pindaro, 42, I-65127 Pescara, Italy.}\email{alessandra.cretarola@unich.it}

\author[B.~Salterini]{Benedetta Salterini\,\orcidlink{0000-0001-6950-4766}}
\address{Benedetta Salterini, Department of Mathematics and Computer Science,
University of Firenze, Viale Morgagni, 67/A, I-50134 Firenze, Italy.}\email{benedetta.salterini@unifi.it}

\title{Indifference pricing of pure endowments in a regime-switching market model}

\maketitle

\begin{abstract}
In this paper, we study the exponential utility indifference pricing of pure endowment policies within a stochastic-factor model for an insurer who also invests in a financial market. Our framework incorporates a hazard rate modeled as an observable diffusion process, while the risky asset price follows a jump-diffusion process driven by a continuous-time finite-state Markov chain, effectively capturing different economic regimes.
Using the classical stochastic control approach based on the Hamilton-Jacobi-Bellman equation, we derive optimal investment strategies with and without the insurance derivative and characterize the indifference price as a classical solution to a linear partial differential equation (PDE). Additionally, we provide a probabilistic representation of the indifference price via an extension of the Feynman-Kac formula and show that it satisfies a suitable backward PDE.
Finally, some numerical experiments are conducted to perform sensitivity analyses, highlighting the impact of key model parameters.
\end{abstract}


\section{Introduction}

Indifference pricing has emerged as a fundamental approach in incomplete markets, particularly for life insurance and pension-related financial products. The exponential utility-based framework provides a robust method for valuing insurance-linked securities by incorporating investor preferences under uncertainty. A significant aspect of this study is the incorporation of regime-switching dynamics into the pricing framework, where the financial market is influenced by a continuous-time finite-state Markov chain.

The utility indifference pricing method, initially proposed by \citet{HodgesNeuberger1989} and refined by \citet{DavisEtAl1993}, has gained much attention in the literature on pricing and hedging contingent claims (see \citet{HendersonHobson2009} for a survey). According to this technique, the indifference seller’s (or insurer’s) price is defined at the level where the issuer of the contract is indifferent between entering the market on its own or selling the claim and entering the market with the collected premium. It can be determined by solving an equation involving two value functions, resulting from the stochastic control problems with and without insurance liabilities. This method has been widely used to price derivatives in incomplete markets and has been successfully applied to insurance contracts (e.g., \citet{LudkowskiYoung2008, Delong2009, EichlerEtAl2017, choi2017indifference, LiangLu2017, CeciEtAl2020}). \\
In recent years, there has been increasing interest in utility-based indifference pricing in regime-switching models. Notably, \citet{ElliottSiu2011} explored utility-based indifference pricing and hedging of derivative securities where financial markets evolve under regime-switching dynamics. %emphasize the impact of economic regimes on financial derivatives pricing
Indifference pricing in regime-switching models is explored within a multi-period stochastic portfolio selection framework, where risk aversion evolves with the regime, leading to time-inconsistent strategies and the computation of utility indifference ask prices via an indifference valuation algorithm, as studied in \citet{pirvu2013utility}.
More recently, in \citet{gyulov2022penalty}, indifference pricing in regime-switching models is studied through a numerical method for American options, formulating the problem as a nonlinear complementarity problem and applying an interior penalty approach for theoretical and computational efficiency. As far as we know, indifference pricing has not been explored in insurance models with regime-switching dynamics, even though such models are increasingly used to capture macroeconomic shifts and structural changes in financial and actuarial contexts.\\
In this paper, we study the indifference pricing problem of pure endowment contracts for an insurer operating in a continuous-time financial market. Specifically, we consider a setting where the price of the risky asset can exhibit jumps and is affected by regime changes, while the hazard rate governing population mortality is stochastic and driven by a general diffusion process. Modeling the hazard rate of individuals as a diffusion process has become standard in the actuarial literature.
The first stochastic mortality model was introduced by \citet{MILPRO}.
Early affine‐diffusion specifications were proposed by \citet{DAHL04,biffis}. Key contributions to stochastic mortality modeling also include the framework of \citet{cairns} for valuing and securitizing mortality risk, the calibration and empirical investigation of affine stochastic intensity models by \citet{LV2008}, and the study of mortality–market dependence in \citet{DC2015}. More recently, \citet{CeciEtAl2020} allowed the mortality intensity to depend on an additional stochastic factor, designed to capture economic and environmental influences.\\ 
A pure endowment is a life insurance contract that provides a fixed payout at a specified future time, provided that the insured person is still alive at that moment. The terminal benefit is predetermined and does not depend on the performance of any traded asset in the financial market.
This type of insurance product has practical applications in various fields, including pension plans and socially responsible insurance policies, where the goal is to ensure a future payout independent of financial market fluctuations. In particular, our modeling framework can be adapted to describe ESG-related (Environmental, Social, and Governance) life insurance products, such as life policies linked to sustainable investment funds or pension plans incorporating social responsibility criteria. In such contracts, the payoff structure remains fixed, aligning with the objectives of financial security and sustainability.
Our framework accounts for three main sources of risk: financial risk arising from asset price fluctuations, economic risk (or regime-switching risk) due to structural shifts in economic conditions, and mortality risk. Inspired by \citet{LudkowskiYoung2008}, we extend previous approaches by incorporating a more sophisticated financial market structure where stock prices exhibit multiple jumps. Indeed, the risky asset price evolution is described by a jump diffusion process where the appreciation rate and the volatility depend on an observable continuous-time, finite-state Markov chain representing the regimes of the economy.
To the best of our knowledge, the indifference pricing of life-insurance liabilities in a Markov-modulated framework, accounting for long-term macroeconomic effects described by a continuous-time Markov chain, possible jumps in the risky asset price dynamics, and a stochastic hazard rate, has not been explored before.\\
Our approach allows to yield explicit results, making it particularly suitable for practical applications in insurance and financial risk management. 
One critical modeling choice in our framework is that the Markov chain affects only the price dynamics of the risky asset and not the interest rate. This assumption is motivated by the observation that economic regime shifts often primarily impact market volatility, liquidity conditions, and risk premiums, while central bank policies tend to stabilize interest rates over medium-term horizons. By adopting this framework, we strike a balance between model complexity and analytical tractability, ensuring that our results remain both rigorous and practically relevant. The choice of modeling assumptions allows us to derive explicit pricing formulas that facilitate risk assessment and decision-making, making our approach particularly relevant for both academic research and industry applications.\\
An additional feature of our model is to take the hazard rate of individuals as a general diffusion process, in order to capture the unexpected changes in mortality. %We are not the first to consider stochastic mortality rates, see e.g. \citet{MILPRO, DAHL04, damo, biffis, LUDYOUNG}. 
%Some studies reveal that historical trends in mortality have provided empirical evidence that mortality rates evolve stochastically over time, see e.g. \citet{cairns}
Indeed, empirical evidence suggests that wars, medical breakthroughs, developments in healthcare and improved lifestyles combine to affect human mortality in a fluctuating and unpredictable manner. The uncertainty given by minuscule and continuous movements of the mortality intensity is usually represented by a Brownian motion, see \citet{cairns} for an overview. As a consequence, it seems reasonable to require that in our setting the exogenous stochastic factor, representing long-term environmental changes, does not affect the mortality intensity; therefore the insurance market remains independent of the financial market. \\
We price the policy through the principle of equivalent utility by comparing the maximal expected utility functions with and without writing the life insurance contract. 
Under exponential utility and using the classical stochastic control approach based on the Hamilton-Jacobi-Bellman (HJB) equation, we describe the optimal investment strategy and show verification results for the value functions of the problems without and with insurance liabilities via classical solutions to a linear PDE and a system of ordinary differential equations (ODEs), see Theorem \ref{thver} and Theorem \ref{thver2}. Further, we characterize the indifference price of the pure endowment in Proposition \ref{thprezzo}. We prove that it solves a propernbackward PDE and we also obtain its probabilistic representation by means of an extended version of the Feynman-Kac formula.
%We also discuss the indifference price for a group of insurance contracts and another kind of mortality-contingent claim.  
Finally, numerical experiments are performed to investigate some features of our model specification, emphasizing the impact of the regime-switching and the randomness effect introduced by the stochastic hazard rate.\\
The remainder of this paper is structured as follows: Section \ref{sec:model} presents the modeling framework, detailing the dynamics of the financial market, mortality intensity, and Markov-modulated regime shifts. Section \ref{sec:form} formulates the optimization problem, while Section \ref{sec:hjb} derives the optimal investment strategies, both with and without the insurance derivative employing dynamic programming. Section \ref{sec:indiff-price} 
characterizes the indifference price for the pure endowment contract
whereas Section \ref{sec:numerical}
provides numerical illustrations. Finally, Section \ref{sec:conclusion} presents the concluding remarks of the paper along with possible directions for future research. All technical proofs and additional auxiliary results are collected in Appendix \ref{app:tech} and Appendix \ref{app:HJB}.

\section{Modeling framework}\label{sec:model}

We consider a complete probability space $(\Omega, \F, \P)$ endowed with a filtration $\mathbb{G}=\{\mathcal{G}_t,\ t \in [0,T]\}$, satisfying the usual conditions of completeness and right continuity, where $T>0$ is a fixed, finite time horizon. %We assume that %$\F_T=\F$ and 
%$\mathcal{G}_0=\{\Omega, \emptyset\}$. 
Specifically, the filtration $\mathbb{G}$ is given by %consists of two subfiltrations, i.e. 
\begin{equation}
	\mathbb{G}= \bF \vee \bF^I,
\end{equation} 
where the filtration $\bF=\{\F_t,\ t \in [0,T]\}$ models the information flow in the financial market  %where the insurer invests 
and $\bF^I=\{\F_t^I,\ t \in [0,T]\}$ contains information about the lifetime of the individual insured. We assume that the subfiltrations $\bF$ and $\bF^I$ are independent.

To describe some possible structural changes in economic conditions, we introduce an irreducible and continuous-time Markov chain $X=\{X_{t}, \ t \in [0,T]\}$ with finite state space $\mathcal{X}=\{ 1,2, \ldots,M \} $, whose transition probabilities satisfy
\begin{equation*}
	\P(X_{t+\delta t}=j | X_{t}=i) = a_{ij} \delta t + o(\delta t), \ i\neq j; \ \ \P(X_{t+\delta t}=i | X_{t}=i) = 1+ a_{ii} \delta t + o(\delta t), 
\end{equation*} when $\delta\longrightarrow 0$, where for each $i \in \mathcal{X}$ we have \begin{equation}\label{propcat}
	a_{ij}\geq 0  \ \ \mbox{for each } i \neq j \ \ \ \mbox{and} \ \ \ a_{ii}=-\sum_{j=1}^{M} a_{ij}.
\end{equation} 
Here, $X_t$ represents the regime of the economy at time $t$, and $M$ the number of regimes.
Let $\textbf{A}=(a_{ij})_{i,j\in\mathcal{X}}$ denote the generating $Q$-matrix of the Markov chain $X$.
 \noindent It is convenient to represent $X$ as a stochastic integral with respect to a Poisson random measure. Following the description of \citet{GOPAL}, for $i,j \in \mathcal{X}$, with $i \neq j$, we denote by $\Delta_{ij}$ the consecutive (with respect to the lexicographic ordering on $\mathcal{X} \times \mathcal{X}$) left-closed right-open intervals of the real line, each having length $a_{ij}$ and define a function $h: \mathcal{X} \times \R \longrightarrow \R^M$ by embedding $\{ 1,2, \ldots,M \}$ into $\R^M$ (identifying $i$ with $e_i \in \R^M$), as follows
 \begin{equation} \label{h}
 h(i,z)=\left \{\begin{array}{ll}
 j-i,	\hspace{0.5cm}	\mbox{if } z\in \Delta_{ij} \\
 0, 	\hspace{1.1cm}	{\rm otherwise}.
 \end{array} \right .
 \end{equation} 
Then, we get
\begin{equation}\label{X}
X_t=X_0 + \int_0^t \int_{\R} h(X_{v-},z) \mathcal{P}(\ud z,\ud v),\quad t \in[0, T], 
\end{equation} 
where the integration is over the interval $(0,t]$ and $\mathcal{P}(\ud z,\ud t)$ is a Poisson random measure with intensity $m(\ud z)\ud t$, with $m(\ud z)$ being the Lebesgue measure on $\R$. Let $\widehat{\mathcal{P}}(\ud z,\ud t)$ be the compensated Poisson random measure, i.e. $\widehat{\mathcal{P}}(\ud z,\ud t) = \mathcal{P}(\ud z,\ud t) - m(\ud z)\ud t$.

In this setting, we consider a financial market consisting of a locally risk-free money market account and one stock as a risky asset. 
The price process $B=\{ B_{t},\ t \in [0,T]\}$ of the locally risk-free asset is described by 
\begin{equation} \label{b}
\ud B_{t}=rB_{t} \ud t, \quad  B_{0}=1,
\end{equation} where $r$ is a positive constant denoting the risk-less interest rate.
The risky asset price process $S=\{S_{t}, \ t\in [0,T]\}$ evolves over time according to the following regime-switching dynamics
\begin{equation}
	\label{s}
	\ud S_{t}=S_{t-}\left(\mu(t,X_{t})\ud t + \sigma(t,X_{t})\ud Z_{t}^S + K_1(t,X_{t-})\ud N^1_t - K_2(t,X_{t-}) \ud N^2_t\right),
\end{equation}
with $ S_{0}=s \in \R^+$, where $\R^+=(0,+\infty)$.
Here, $Z^S=\{Z_{t}^S,\ t \in [0,T]\}$ is a standard Brownian motion independent of $X$ and $N^1=\{N_t^1,\ t \in [0,T]\}$ and $N^2=\{N_t^2,\ t \in [0,T]\}$ are independent Poisson processes defined on $(\Omega,\mathcal{F},\P;\bF)$. 
Furthermore, we suppose that $N^1$, $N^2$ are independent of $Z^S$ and $X$ and that the $\bF$-intensities of $N^1$ and $N^2$ are 
positive deterministic functions $\Theta_1:[0,T] \longrightarrow \R^+$ and $\Theta_2:[0,T] \longrightarrow \R^+$, respectively. 
The coefficients $\mu: [0,T] \times \mathcal{X} \longrightarrow \R^+$ and $\sigma: [0,T] \times \mathcal{X} \longrightarrow \Rp$ are measurable functions which model the appreciation rate and the volatility of the stock, respectively, such that $\mu(t,i)>r$, for all $(t,i) \in [0,T]\times \mathcal X$ and 
\begin{equation}\label{ipsmg}
	\int_0^T  \left( \mu(t,X_t) + \sigma^2(t,X_t) %+ K_1^2(t,X_{t-})\Theta_1(t) + K_2^2(t,X_{t-}) \Theta_2(t) 
	\right) 
	\ud t < \infty \quad \P \mbox{-a.s.}.
\end{equation}
Moreover, $K_1:[0,T] \times \mathcal{X} \longrightarrow \R^+$ and $K_2:[0,T] \times \mathcal{X} \longrightarrow \R^+$ are measurable functions such that $K_l(t,i)>0$, $l=1,2$ and $K_2(t,i)<1$, for every $(t,i) \in [0,T] \times \mathcal{X}$. From \eqref{X} and \eqref{s} it is clear that the pair $(S,X)$ is an $(\mathbb{F},\P)$-Markov process. The main motivation for introducing a regime-switching behavior is to have a model capable of describing the risky asset price dynamics under different market conditions. Specifically, this model provides a rich framework to capture both smooth price movements and abrupt changes in risky asset prices, influenced by market regimes and external shocks. The term \(K_1(t, X_{t-}) \ud N_t^1\) introduces upward jumps in the asset price, driven by the Poisson process \(N^1\) with intensity potentially depending on \(X\). Economically, these jumps could represent sudden favorable events such as positive earnings announcements, technological breakthroughs, or market expansions. On the other hand, the term \(-K_2(t, X_{t-}) \ud N_t^2\) models downward jumps in the asset price, driven by a separate Poisson process \(N^2\). These could represent adverse events such as unexpected losses, regulatory shocks, or macroeconomic downturns.
\begin{remark}
\label{Ssol}
By the Doléans-Dade exponential formula, condition $K_2(t,i) < 1$ allows us to write
\begin{equation}
S_t=s e^{L_t}, \quad t \in [0,T],
\end{equation} 
where the logreturn process $L=\{L_{t}, \ t \in [0,T] \}$ is given by
\begin{equation}
\begin{split}
\ud L_t & = \big( \mu(t,X_t)-\frac{1}{2} \sigma^2(t,X_t) \big) \ud t + \sigma(t,X_t) \ud Z_t^S + \ln(1+K_1(t,X_{t-})) \ud N^1_t \\
& \qquad + \ln(1-K_2(t,X_{t-})) \ud N^2_t, \quad L_0=0.
\end{split}
\end{equation}
%with $L_0=0$.
\end{remark}

\begin{proposition}
If we assume that 
\begin{equation}\label{ipN}
	\int_0^T  \left( K_1^2(t,X_{t-})\Theta_1(t) + K_2^2(t,X_{t-}) \Theta_2(t) \right) \ud t < \infty \quad \P \mbox{-a.s.},
\end{equation}
then the process $S$ is an $\bF$-semimartingale with decomposition
\begin{equation}\label{ssmg}
	S_t=s+A_t^S+M_t^S,
\end{equation} 
where $A^S=\{A_{t}^S, \ t \in [0,T] \}$ defined as
\begin{equation}
	A_t^S = \int_{0}^t S_{v-} \left(\mu(v,X_{v-})+ K_1(v,X_{v-})\Theta_1(v) + K_2(v,X_{v-}) \Theta_2(v)\right)\ud v,
\end{equation} 
is an $\R$-valued process with finite variation paths and $A_0^S=0$, while $M^S=\{M_{t}^S, \ t \in [0,T] \}$ given by 
\begin{equation}
\begin{split}
	 M_t^S&=\int_{0}^t   S_v \sigma(v,X_v) \ud Z_v^S + \int_{0}^{t} S_{v-} K_1(v,X_{v-}) \{ \ud N^1_v -\Theta_1(v) \ud v \} \\
     & \qquad - \int_{0}^{t} S_{v-} K_2(v,X_{v-}) \{ \ud N^2_v -\Theta_2(v) \ud v \}
     \end{split}
\end{equation} 
is an 
%locally square-integrable 
$\bF$-local martingale with $M_0^S=0$. 
\begin{proof}
The proof follows directly from \eqref{s}.
\end{proof}
\end{proposition}
\noindent Now, we consider an individual %aged $l$ at time $0$ 
to be insured and a stochastic model for the mortality of the equivalent age cohort of the population.
%to incorporate stochastic movements of hazard rate, 
We assume that the hazard rate (or mortality intensity) is governed by a diffusion process.
Stochastic intensity models for mortality have been proposed in the literature by \citet{MILPRO}, with further contributions by \citet{DAHL04,biffis,cairns,LV2008,DC2015}, and, more recently, by \citet{CeciEtAl2020}.
%\blublu{Stochastic intensity models for mortality were introduced by \citet{MILPRO}, with affine‐diffusion specifications by \citet{DAHL04,biffis}, and further developments in \citet{cairns,LV2008,DC2015}, and most recently in \citet{CeciEtAl2020}.}
Precisely, we describe the mortality intensity as a stochastic process $\Lambda = \{\lambda_t,\ t \in [0,T]\}$ that is given by the following stochastic differential equation (in short SDE) 
\begin{equation}
	\label{y}
\ud\lambda_{t}=\lambda_t\left(b(t,\lambda_{t}) \ud t+ c(t,\lambda_{t}) \ud Z_{t}^{\Lambda}\right), \quad  \lambda_{0}=\lambda \in \R^+.
\end{equation}
Here, $Z^{\Lambda}=\{Z^{\Lambda}_{t}, \ t \in [0,T]\}$ is an additional standard Brownian motion on $(\Omega,\F,\P;\bF^I)$. 
Moreover, $b: [0,T] \times \R \longrightarrow \R$ and $c: [0,T] \times \R \longrightarrow \R$ are two measurable functions such that a unique strong solution to \eqref{y} exists and the following conditions hold
\begin{equation}\label{eq:exp_coeff}
\esp{\int_0^T |b(t,\lambda_t) \lambda_t| \ud t + \int_0^T c(t,\lambda_t)^2 \lambda_t^2 \ud t}< \infty,
\end{equation}
\begin{equation}\label{cond:sup}
\sup_{t \in [0,T]}\esp{\lambda_t^2} < \infty.
\end{equation}
These conditions are satisfied if, for instance, the coefficients of the SDE \eqref{y} fulfill the classical Lipschitz and sublinear growth conditions, see e.g. 
\citet{GIHMSKOR}.

We observe that, %due to the presence of many exogenous factors, 
the mortality rate of the insured is generally different from that of its age cohort. However, to keep the framework tractable we consider individuals subjected to the same stochastic hazard rate, as e.g. in \citet{LudkowskiYoung2008}.
%\end{remark} 

Let $\tau$ be a nonnegative random variable on $(\Omega,\mathcal{F},\P)$ which represents the remaining lifetime of the given individual of the reference population with mortality rate $\Lambda$. 
Denote by $D=\{D_{t}, \ t \in [0,T] \}$ the death indicator process associated with $\tau$ by setting $D_{t}:=\I_{\{\tau \leq t \}}$, for every $t \in [0,T]$.
Clearly, $D$ is an $\bF^I$-adapted process.

\section{The indifference pricing problem formulation}\label{sec:form}

Now, we assume that the insurer issues pure endowment policy, which is a life insurance contract 
where the insured receives a specified amount (the endowment) if she/he survives until the end of a predetermined term. If the policyholder dies before the term ends, no benefit is paid.
In particular, we consider a pure endowment contract with maturity of $T$ years, which pays a fixed amount if the policyholder is still alive. Then, the associated payoff is given by the random variable 
	\begin{equation} \label{payoffPE}
		G_{T}:= K \I_{\{\tau > T \}} = K(1-D_T),
	\end{equation} 
where $K$ is a positive constant.
The goal is to evaluate the pure endowment policy with payoff given by \eqref{payoffPE} in the Markov-modulated model outlined in Section $2$. 
Since the financial market consists of two primary securities and several sources of random shocks due to mortality events and structural changes in economic conditions, it turns out to be incomplete.

Therefore, we apply the indifference pricing approach assuming that the insurer preferences towards the risk are given by an exponential utility function of the form \begin{equation}
\label{uesp}
u(w)=-e^{-\alpha w}, \quad  w \in \R,
\end{equation} 
where $\alpha$ is a positive parameter which measures the absolute risk aversion.
Since there is no universally accepted utility function for decision makers' preferences, we adopt the exponential utility in \eqref{uesp} as a representative case: it preserves the essential features of indifference pricing while allowing for more explicit computations.
As we will see, this choice removes wealth from the HJB’s second-order term and thus leads to a linear backward PDE (two decoupled linear PDEs in the case with the insurance derivative) and an associated ODE system. By contrast, for any other non-CARA preference the HJB remains fully nonlinear and typically admits no classical solution.\\
%This CARA choice is standard but also essential: 
%As we will see in the sequel, this choice removes the dependence on wealth from the HJB's second-order term and hence yields to linear backward PDEs (and associated ODE system); %the same trick also applies to logarithmic utility, which likewise admits a separation of the wealth variable, whereas for any other non-CARA/CRRA 
%while, for any other preference the HJB remains fully nonlinear and admits in general no classical solution.}\\
In the underlying financial market, the insurer starts out with an initial wealth $w$, and then proceeds to trade dynamically among the risky asset and the locally risk-free asset, following a self-financing strategy. Let $\Pi=\{ \Pi_{t}, \ t \in [0,T]\}$ be the total amount of wealth invested in the stock, with the remainder of wealth in the money market account. 
The insurer is also allowed to short-sell and to borrow/lend any infinitesimal amount, so that $\Pi_t \in \R$, for each $t \in [0,T]$. 
Finally, we focus exclusively on self-financing strategies: the insurer reinvests only the profits generated by her/his core business into the available assets, without withdrawing any funds or injecting external capital.
Precisely, given an initial wealth $w \in \R^+$, the insurer wealth process $\{W_t^\Pi,\ t \in [0,T]\}$ associated to a given strategy $\Pi$ evolves over time as
\begin{equation}\label{wealth}
\begin{split}
\ud W_{t}^\Pi &= \Pi_{t} \frac{\ud S_{t}}{S_{t-}} + (W_{t}^\Pi-\Pi_{t})\frac{\ud B_{t}}{B_{t}} 
\\ &= \left(rW_{t}^\pi+\Pi_{t}\left(\mu(t,X_{t})-r \right)\right) \ud t \\
& \qquad + \Pi_{t}\sigma(t,X_{t})\ud Z_{t}^S + \Pi_{t} \Big( K_1(t,X_{t-}) \ud N^1_t - K_2(t,X_{t-}) \ud N_t^2 \Big), 
\end{split}
\end{equation} 
with $W_{0}^\Pi=w \in \R^+$.
\begin{remark}
It can be checked that the solution to the SDE \eqref{wealth} is given by
\begin{equation}\begin{split}\label{def:solwealth}
W_t^{\Pi} &= W_0^{\Pi} e^{rt} + \int_0^t e^{r(t-s)} \Pi_s(\mu(s,X_s)-r) \ud s + \int_0^t e^{r(t-s)}\Pi_s \sigma(s,X_s) \ud Z_s^S \\ & \quad + \int_0^t e^{r(t-s)} \Pi_s \Big( K_1(s,X_{s-}) \ud N^1_s - K_2(s,X_{s-}) \ud N_s^2 \Big) , 
\end{split}\end{equation} with $W_{0}^\Pi=w$. Indeed, if we set
\begin{equation}\begin{split}
U_t^{\Pi} &:= w + \int_0^t e^{-rs} \Pi_s(\mu(s,X_s)-r) \ud s + \int_0^t e^{-rs}\Pi_s \sigma(s,X_s) \ud Z_s^S \\ & \quad + \int_0^t e^{-rs} \Pi_s \Big( K_1(s,X_{s-}) \ud N^1_s - K_2(s,X_{s-}) \ud N_s^2 \Big), \quad t \in [0,T],
\end{split}\end{equation}
then by \eqref{def:solwealth} we have \(W_t^{\Pi}=e^{rt}U_t^{\Pi}\).
Applying It\^o's product rule gives
\begin{align}
\ud W_t^{\Pi}
&= \mathrm{d}\bigl(e^{rt}U_t^{\Pi}\bigr)
= r\,e^{rt}U_t^{\Pi}\,\ud t + e^{rt}\,\ud U_t^{\Pi} \\[4pt]
&= r\,W_t^{\Pi}\,\ud t
  + e^{rt}\,\biggl[
      e^{-r t}\,\Pi_t(\mu(t,X_t)-r)\,\ud t
      + e^{-r t}\,\Pi_t\,\sigma(t,X_t)\,\ud Z_t^S \notag\\
&\quad\hphantom{{}={}}
      + e^{-r t}\,\Pi_{t}\bigl(K_1(t,X_{t^-})\,\ud N^1_t
                             -K_2(t,X_{t^-})\,\ud N^2_t\bigr)
    \biggr]\\[4pt]
&= \bigl(r\,W_t^{\Pi} + \Pi_t(\mu(t,X_t)-r)\bigr)\,\ud t
  + \Pi_t\,\sigma(t,X_t)\,\ud Z_t^S\\[4pt]
  & \quad \hphantom{{}={}} + \Pi_{t}\bigl(K_1(t,X_{t^-})\,\ud N^1_t
                 -K_2(t,X_{t^-})\,\ud N^2_t\bigr),
\end{align}
which is exactly \eqref{wealth}. Furthermore, by standard existence and uniqueness theorems for linear SDEs with jumps, the explicit solution \eqref{def:solwealth} is the unique  càdlàg solution of \eqref{wealth}.

\end{remark}

\noindent In order to characterize the indifference price of the pure endowment (see Definition \ref{indiff-price} below), we introduce two optimal investment problems, with and without insurance liabilities. We start by defining 
the class of admissible strategies.
\begin{definition}
An {\em admissible} strategy is a self-financing portfolio identified by an $\R$-valued $\mathbb{G}$-predictable process $\Pi=\{ \Pi_{t}, \ t \in [0,T]\}$ 
such that 
\begin{equation} \label{int_ammiss1}
\esp{\int_0^T |\Pi_t| (\mu(t,X_t)-r) \ud t} < \infty, \qquad \esp{\int_0^T\Pi_t^2 \sigma^2(t,X_t)\ud t} < \infty,
\end{equation}
\begin{equation} \label{int_ammiss3}
\esp{\int_0^T |\Pi_t| \Big( K_1(t,X_{t-})\Theta_1(t) + K_2(t,X_{t-}) \Theta_2(t) \Big) \ud t}< \infty.
\end{equation}
We denote by $\mathcal{A}$ the set of $\mathbb{G}$-admissible strategies. Whenever the controls are restricted to the time interval $[t,T]$, we will use the notation $\mathcal{A}_t$.
\end{definition}
\noindent Now, we assume that the following assumptions are in force throughout the paper. 
\begin{ass}\label{int_condit}
\begin{itemize}
\item[]
    \item[(i)] There exist three positive constants $M_1$, $M_2$ and $K$ such that
    \begin{equation}
        \Theta_1(t) \le M_1, \quad \Theta_2(t) \le M_2, \quad K_1(t,i) \le K, \quad \mbox{for\ every}\ (t,i) \in [0,T] \times \mathcal X.
    \end{equation}
    \item[(ii)] There is a constant $C >0$ such that $\frac{\mu(t,i)-r}{\sigma(t,i)} \le C$, for every $(t,i) \in [0,T] \times \mathcal X$.
\end{itemize}
\end{ass}
\noindent In particular, Assumption \ref{int_condit}$(i)$ provides a sufficient condition for a strategy $\Pi$ to be admissible as shown in the next result.
%Notice that assuming $(i)$ implies that condition \eqref{ipN} holds true and therefore it ensures the semimartingale structure for the traded asset $S$.
\begin{proposition}\label{suff}
Let $\Pi=\{ \Pi_{t}, \ t \in [0,T]\}$ be a $\mathbb G$-predictable strategy with values in $\R$. Assume there exists a square-integrable function $\eta:[0,T] \times \mathcal X \to \R^+$ such that
\begin{equation}\label{eq:pi-eta}
  \left| \Pi_t \right| \le \eta (t,X_t), \quad t \in [0,T],\quad \P-a.s.
\end{equation}
and
\begin{equation}\label{eta_cond}
    \int_0^T \eta(s,i)\left((\mu(s,i) - r)+\eta(s,i)\sigma^2(s,i)\right) \ud s < \infty, \quad \forall i \in \mathcal X.
\end{equation}
Then, $\Pi$ is an admissible strategy, i.e. $\Pi \in \mathcal A$.
\end{proposition}
\begin{proof}
We note that by \eqref{eta_cond}, we have
\begin{align}
    & \esp{ \int_0^T |\Pi_s| \left( (\mu(s,X_s) - r)+\Pi_s\sigma^2(s,X_s)\right) \ud s}\\
    & \qquad \le \esp{ \int_0^T \eta(s,X_s)\left( (\mu(s,X_s) - r)+\eta(s,X_s)\sigma^2(s,X_s)\right) \ud s}\\
    & \qquad \le \max_{i=1,\ldots,M}\int_0^T \eta(s,i)\left( (\mu(s,i) - r)+\eta(s,i)\sigma^2(s,i)\right) \ud s < \infty.
\end{align}
Finally, in view of Assumption \ref{int_condit}(i), condition  \eqref{int_ammiss3} is satisfied and this concludes the proof.
\end{proof}
\noindent Firstly, we consider the case where the insurer simply invests her/his wealth in the financial market, without writing the mortality contingent claim. Then, the goal is the following. \begin{problem}
	\label{prbesp}
To maximize the expected utility of its terminal wealth, i.e. to solve \begin{equation} 
	\sup_{\Pi \in \mathcal{A}} \mathbb{E} \Big{[} -e^{-\alpha W_{T}^\Pi} \Big{]} .
	\end{equation}  
\end{problem} 
\noindent Let $(t,w,i) \in [0,T] \times \R \times \mathcal{X}$. In a dynamic framework, we define the corresponding value function $\bar{V}$ by 
\begin{equation}
\label{Vbaruesp}
\bar{V}(t,w,i):=\sup_{\Pi \in \mathcal{A}_t} \mathbb{E} _{t,w,i} \Big{[} -e^{-\alpha W_{T}^{\Pi}(t,w)} \Big{]},
\end{equation}
where $ \mathbb{E} _{t,w,i}[\cdot]$ denotes the conditional expectation given $W_{t}^\Pi=w$ and $X_{t}=i$, and $\{W_{s}^{\Pi}(t,w),\ s \in [t,T]\}$ stands for the solution to equation \eqref{wealth} with initial condition $W_t^\Pi=w$.
%\begin{remark}
Note that, since the coefficients $\mu$, $\sigma$, $K_1$ and $K_2$ only depend on $t$ and $i$, it is possible to absorb the stock price in the wealth and therefore to remove the variable corresponding to $S$.
%\end{remark} 	

Now, we suppose that the insurer invests her/his wealth in the market, having the opportunity to write a pure endowment contract with payoff given in \eqref{payoffPE}.
In this case, the goal of the insurer is the following.
\begin{problem}
	\label{prbcesp}
To maximize the expected utility of its terminal wealth, i.e. to solve \begin{equation} 
	\sup_{\Pi \in \mathcal{A}} \mathbb{E} \Big{[} -e^{-\alpha (W_{T}^\Pi-G_T)} \Big{]},
	\end{equation} 
	where $G_T$ is defined in \eqref{payoffPE}. 
\end{problem} 
\noindent Let $(t,w,\lambda,i) \in [0,T] \times \R \times \Rp \times \mathcal{X}$. We define the corresponding value function $V$ as 
\begin{equation}
\label{Vuesp}
V(t,w,\lambda,i):=\sup_{\Pi \in \mathcal{A}_t} \mathbb{E} _{t,w,\lambda,i} \Big{[} -e^{-\alpha (W_{T}^\Pi(t,w)-G_T)} \Big{]},
\end{equation} 
where $ \mathbb{E}_{t,w,\lambda,i} $ denotes the conditional expectation given $W_{t}^\Pi=w$, $\lambda_{t}=\lambda$ and $X_{t}=i$ and we implicitly condition on $G_{t}=K$.
\begin{definition}\label{indiff-price}
Given $W_{t}=w$, $\lambda_{t}=\lambda$ and $X_{t}=i$, the {\em indifference price}  or {\em reservation price} of the insurer related to the pure endowment contract is defined at any time $t \in [0,T]$ as the $\mathbb{G}$-adapted process  $P=\{P_t,\ t \in [0,T]\}$ implicit solution to the equation 
	\begin{equation}
	\label{przz}
		\bar{V}(t,w,i)=V(t,w+P_{t},\lambda,i).
	\end{equation} 
    \end{definition}
    \noindent This means that starting at time $t$ with wealth $w$, hazard rate $\lambda$ and scenario $i$, the insurer has the same maximal utility from selling the pure endowment contract for $P_t$ at time $t$ and solely trading on $(t,T]$ without writing the contract.
To find $P$ one has to characterize the value functions given in \eqref{Vbaruesp} and \eqref{Vuesp}, corresponding to the 
stochastic control problems without and with the insurance derivative.
\begin{remark}
We note that the control $\Pi=0$ is admissible and such that
$$
\mathbb{E} _{t,w,i}\left[e^{-\alpha W_T^0(t,w)}\right] < \infty, 
$$
for each $(t,w,i) \in [0,T] \times \R \times \mathcal{X}$.
$$
\mathbb{E} _{t,w,\lambda,i}\left[e^{-\alpha (W_T^0(t,w)-G_T)}\right] < \infty, 
$$
for each $(t,w,\lambda,i) \in [0,T] \times \R \times \Rp \times \mathcal{X}$. This implies that
\begin{equation}
    \esssup_{\Pi \in \mathcal A_t}\esp{-e^{-\alpha W_{T}^\Pi}} > - \infty, \quad \esssup_{\Pi \in \mathcal A_t}\esp{-e^{-\alpha (W_{T}^\Pi-G_T)}} > - \infty,\quad \P-\mbox{a.s.}, \ t \in [0,T],
\end{equation}
and as a consequence that
\begin{equation}
    \sup_{\Pi \in \mathcal A}\esp{-e^{-\alpha W_{T}^\Pi}} > -\infty, \quad \sup_{\Pi \in \mathcal A}\esp{-e^{-\alpha (W_{T}^\Pi-G_T)}} > -\infty.
\end{equation}
\end{remark}

\section{The optimal investment problems}\label{sec:hjb}

In this section, 
applying the classical stochastic control approach based on the Hamilton-Jacobi-Bellman (HJB) equation,
we characterize the optimal investment strategies and provide verification results for the value functions $\bar V$ and $V$ given in \eqref{Vbaruesp} and \eqref{Vuesp}, respectively.
%In the sequel, we assume that the functions $K_1$ and $K_2$ are \aC{continuous}.

\subsection{The pure investment problem}

Firstly, we consider the case where the insurer simply invests in the underlying financial market, %does not underwrite any claims and 
so the corresponding value function $\bar{V}$ is given by \eqref{Vbaruesp}. 
Let us consider the HJB equation with final condition that the value function $\bar{V}$ is expected to solve, if sufficiently smooth:
\begin{equation}\label{VbarHJB}
\left \{ 
\begin{array}{ll} 
%\begin{align} 
\sup_{\Pi \in \R} \bar {\mathcal L}_i^\Pi \bar{V}(t,w,i) = 0, \quad & \forall (t,w,i) \in [0,T) \times \R \times \mathcal X,\\
\bar{V}(T,w,i)  =-e^{-\alpha w}, \quad & \forall (w,i) \in \R \times \mathcal X,
%\end{align}
\end{array} 
\right.
\end{equation}
where $\bar {\mathcal L}_i^\Pi$ denotes the Markov generator of $(W^\Pi,X)$ associated with a constant control $\Pi \in \R$. 
\begin{definition}
    The set $\mathcal D(\bar {\mathcal L}_i^\Pi)$ denotes the class of functions $f (\cdot,\cdot,i) \in C^1([0,T]) \times C^2(\R)$, for each $i \in \mathcal X$, such that for every constant $\Pi \in \R$, we have
    \begin{align} \label{domain1}
     \mathbb{E} \bigg[ \int_{0}^{T} \Big( \sigma(v,X_{v})\Pi \pd{f}{w}(v,W_{v}^\Pi,X_{v}) \Big)^{2} \ud v \bigg] < \infty,
    \end{align}
    and 
    \begin{align}\label{domain2}
        \esp{\int_0^T \int_{\R}  \left|f\big{(}v,W_{v}^\Pi,X_{v-\!}+h(X_{v-\!},z) \big{)} - f(v,W_{v}^\Pi,X_{v-\!})\right| m(\ud z) \ud v} < \infty,\\
        \esp{\int_0^T \left|f\big{(}v,W_{v-}^\Pi+ \Pi K_1(v,X_{v-}),X_{v-}) \big{)} - f(v,W_{v-}^\Pi,X_{v-})\right|  \Theta_1(v) \ud v } < \infty,\\
        \esp{\int_0^T \left|f\big{(}v,W_{v-}^\Pi- \Pi K_2(v,X_{v-}),X_{v-}) \big{)} - f(v,W_{v-}^\Pi,X_{v-})\right|  \Theta_2(v) \ud v} < \infty.
    \end{align}
\end{definition}

\begin{lemma}
	\label{wyx}
The stochastic process $(W^\Pi,X)$ is a Markov process on $(\Omega, \mathcal{F},\P;\mathbb{G})$, with infinitesimal generator $\bar {\mathcal L}_i^\Pi$ for all constant strategies $\Pi \in \R$, given by
 \begin{equation}
		\label{margen3}
		\begin{split}
			&\bar{\mathcal{L}}_i^\Pi f (t,w,i) = \pd{f}{t}(t,w,i)\\
			&  + \big[ rw + (\mu(t,i)-r)\Pi \big]\pd{f}{w}(t,w,i) + \frac{1}{2} \sigma^{2}(t,i)\Pi^{2} \pds{f}{w}(t,w,i) + \sum_{j \in \mathcal{X}}\! a_{ij} f(t,w,j) \\ &  + \Theta_1(t) \big{\{} \!\bar{V}(t,w+\Pi K_{1\!}(t,i),i) \! -\!\bar{V}(t,w,i) \! \big{\}} + \Theta_2(t) \big{\{} \! \bar{V}(t,w-\Pi K_2(t,i),i) \! -\! \bar{V}(t,w,i) \!\big{\}},
	\end{split}
	\end{equation}
	for every $i \in \mathcal X$. The domain of the generator $\bar {\mathcal L}_i^\Pi$ is $\mathcal D(\bar {\mathcal L}_i^\Pi)$, for each $i \in \mathcal X$.
    \end{lemma}
    \noindent In Appendix \ref{app:tech}, we provide the proof for the Markov property of the over-refined process $(W^\Pi,X,\Lambda)$ stated in Lemma \ref{wysx} below; one can use the same arguments to show the result in this simpler case.
\begin{remark}
Since the pair $(W^\Pi,X)$ is a Markov process,
any Markovian control is of the form $\Pi_t = \Pi(t,W_t^\Pi,X_t)$. The generator $\bar {\mathcal L}_i^\Pi f(t,w,i)$ associated to a general Markovian strategy can be easily obtained by replacing $\Pi$ with $\Pi(t, w, i)$ in \eqref{margen3}.
\end{remark}

\noindent Now, we consider the ansatz 
$\bar{V}(t,w,i)=-e^{-w \alpha e^{r(T-t)}} \varphi(t,i)$, 
with $(t,w,i) \in [0,T] \times \R \times \mathcal X$, for a suitable function $\varphi$, which is motivated by the following result.
\begin{proposition}\label{prop:value}
Assume that there exists a unique function $\varphi(\cdot,i)$, for each $i \in \mathcal X$, % :[0,T] \times \mathcal X \to \R^$
solution to the following Cauchy problem:
\begin{equation}
\label{PCvarphi}
\left\{
\begin{array}{ll}
\pd{\varphi}{t}(t,i) = H(t,\varphi(t,i)),\quad & t \in [0,T),\\
\varphi (T,i)=1, & %\forall i \in \mathcal X,
\end{array}
\right.
\end{equation}
where 
\begin{equation}\label{eq:H}
	H(t,\varphi(t,i))= - \sum_{j \in \mathcal{X}} \varphi(t,j) a_{ij} - \varphi(t,i) \inf_{\Pi \in \R} \bar{\Psi}^\Pi(t,i),
\end{equation}
with the function $\bar{\Psi}^\Pi:[0,T] \times \mathcal X \to \R$ defined by 
\begin{equation}
\begin{split} \label{psibar}
\bar{\Psi}^\Pi(t,i) = & - \alpha e^{r(T-t)} (\mu(t,i)-r) \Pi + \frac{1}{2} \alpha^2 e^{2r(T-t)} \sigma^{2}(t,i) \Pi^2 \\ 
& + \Theta_1(t) \big( e^{-\alpha \Pi K_1(t,i) e^{r(T-t)}} -1 \big) + \Theta_2(t) \left( e^{\alpha \Pi K_2(t,i) e^{r(T-t)}} -1 \right).
\end{split}
\end{equation} 
Then, the function 
\begin{equation}\label{eq:valuefunction1}
\bar{V}(t,w,i)=-e^{-w \alpha e^{r(T-t)}} \varphi(t,i), 
\end{equation}
solves the HJB problem given in \eqref{VbarHJB}.
\end{proposition}

\begin{proof}
From the expression \eqref{eq:valuefunction1}, we can easily verify that the original HJB problem given in \eqref{VbarHJB} reads as follows
\begin{equation}\label{eq:varphi}
\begin{split}
& \pd{\varphi}{t}(t,i)  + \sum_{j \in \mathcal{X}} \varphi(t,j) a_{ij}  + \inf_{\Pi \in \R} \!  \Big{\{}  -\alpha e^{r(T-t)}\varphi(t,i) (\mu(t,i)-r) \Pi \\
& \quad + \frac{1}{2} \alpha^2 e^{2r(T-t)} \varphi(t,i)\sigma^{2}(t,i) \Pi^2
+ \varphi(t,i) \Theta_1(t) \big( e^{-\alpha \Pi K_1(t,i)e^{r(T-t)}} -1 \big)    \\
& \quad \quad + \varphi(t,i) \Theta_2(t) \big( e^{\alpha \Pi K_2(t,i)e^{r(T-t)}} -1 \big) \Big{\}}=0,
\end{split}
\end{equation}
for every $t \in [0,T)$,
with final condition $\varphi(T,i)=1$, for all $i \in \mathcal X$.
Thus, if we define the function $\bar{\Psi}^\Pi$ by means of expression \eqref{psibar},
equation \eqref{eq:varphi} can be written as
\begin{equation}
\pd{\varphi}{t}(t,i) + \sum_{j \in \mathcal{X}} \varphi(t,j) a_{ij} + \varphi(t,i) \inf_{\Pi \in \R} \bar{\Psi}^\Pi(t,i) =0
\end{equation} 
and we find out the problem \eqref{PCvarphi}.
\end{proof}

\noindent The previous result suggests to focus on the minimization  of the function \eqref{psibar}, that is the aim of the next subsection.

\subsubsection{Optimal investment strategy without the insurance derivative}
%In order to characterize the optimal portfolio for an insurer who does not write the insurance derivative, 
Now, we study the following minimization problem 
\begin{equation}	\label{min-prob}
\inf_{\Pi \in \R}{ \bar{\Psi}^\Pi}(t,i),
\end{equation} 
where the function $\bar{\Psi}^\Pi$ is introduced in \eqref{psibar}. 

\begin{proposition}\label{opwi}
The following equation 
\begin{equation}
\begin{split}\label{eqh}
 & \sigma^2(t,i) \alpha e^{r(T-t)} \Pi  - (\mu(t,i) - r)\\
 & \qquad = K_1(t,i) \Theta_1(t) e^{-\alpha \Pi K_1(t,i) e^{r(T-t)}} \!- K_2(t,i) \Theta_2(t) e^{\alpha \Pi K_2(t,i) e^{r(T-t)}}
\end{split}
\end{equation} 
admits at least a solution $\widehat \Pi(t,i)$ in $\R$ for any $(t,i) \in [0,T] \times \mathcal X$ and the minimization problem \eqref{min-prob} has a unique solution $\Pi^*(t,i) = \widehat \Pi(t,i)$, for all $(t,i) \in [0,T] \times \mathcal X$. 
\end{proposition}

\begin{proof}
Firstly, we observe that $\bar{\Psi}^\Pi(t,i)$ is continuous with respect to $\Pi \in \R$, for every $(t,i) \in [0,T] \times \mathcal{X}$ and has continuous first and second order derivatives with respect to $\Pi \in \R$, which are respectively given by 
\begin{align}
& \pd{\bar{\Psi}^\Pi}{\Pi}(t,i)  = -\alpha e^{r(T-t)}(\mu(t,i) - r) + \sigma^2(t,i) \alpha^2 e^{2r(T-t)} \Pi \\
& \quad - \alpha e^{r(T-t)}K_1(t,i) \Theta_1(t) e^{-\alpha \Pi K_1(t,i) e^{r(T-t)}}+ \alpha e^{r(T-t)} K_2(t,i) \Theta_2(t) e^{\alpha \Pi K_2(t,i) e^{r(T-t)}},
\\  & \pds{\bar{\Psi}^\Pi}{\Pi}(t,i) = \alpha^2 e^{2r(T-t)}\sigma^2(t,i) + \alpha^2 e^{2r(T-t)}K_1^2(t,i) \Theta_1(t) e^{-\alpha \Pi K_1(t,i) e^{r(T-t)}} \\ &\quad  + \alpha^2 e^{2r(T-t)} K_2^2(t,i) \Theta_2(t) e^{\alpha \Pi K_2(t,i) e^{r(T-t)}}.
\end{align} 
Note that these derivatives are well defined and %the second order derivative is strictly positive, i.e. 
$\pds{\bar{\Psi}^\Pi}{\Pi}(t,i)>0$, for every $(t,i) \in [0,T] \times \mathcal{X}$; therefore, the function $\bar{\Psi}^\Pi(t,i)$ is strictly convex in $\Pi \in \R$.
Moreover, it is easy to check that, for any $(t,i) \in [0,T] \times \mathcal X$, we have 
\begin{equation}
	\lim_{\Pi \longrightarrow +\infty} \pd{\bar{\Psi}^\Pi}{\Pi}(t,i) \longrightarrow + \infty,\qquad \lim_{\Pi \longrightarrow -\infty} \pd{\bar{\Psi}^\Pi}{\Pi}(t,i) \longrightarrow - \infty. 
\end{equation} 
As a consequence, being $\pd{\bar{\Psi}^\Pi}{\Pi}(t,i)$ a continuous function in $\Pi \in \R$, there exists $\widehat \Pi (t,i)\in \R$ such that $\pd{\bar{\Psi}^\Pi}{\Pi}(t,i)=0$, for every $(t,i) \in [0,T] \times \mathcal{X}$, that is, \eqref{eqh} is satisfied. Since the function $\bar{\Psi}^\Pi(t,i)$ is strictly convex, the stationary point $\widehat \Pi (t,i)\in \R$ is unique and provides the unique minimizer $\Pi^*(t,i)= \widehat \Pi(t,i)$ on $\R$.
\end{proof}

\begin{remark}
We point out that $\Pi^*=\Pi^*(t,i)$, i.e. the solution of the problem \eqref{min-prob} depends on time and on the Markov chain, since it solves equation \eqref{eqh}. This means that the optimal investment strategy evolves over time and changes according to the different economic regimes. Moreover, %observing \eqref{eqh}, 
we note that $\Pi^*$ does not depend on wealth, as usually happens when the investor's preferences are described by an exponential utility function.
\end{remark}
\begin{proposition}[Properties of $\Pi^*$]\label{properties}
%If $\mu(t,i)>r$, for every $(t,i) \in [0,T] \times \mathcal{X}$, 
The following condition is satisfied
\begin{equation}
  \min \left\{  0,\frac{\ln \left( \frac{\mu(t,i)-r}{M_2}\right)}{\alpha e^{r(T-t)}}  \right\} \leq  \Pi^*(t,i) \leq \frac{\mu(t,i) - r + KM_1}{\sigma^2(t,i)\alpha e^{r(T-t)}},
\end{equation}
for all $(t,i) \in [0,T] \times \mathcal X$, where  $K, M_1 \in \R^+$ are the constants limiting the functions $K_1$ and $\Theta_1$, respectively.
\end{proposition}

\begin{proof}
By Proposition \ref{opwi} (we omit the dependence in $\Pi^*$ on $(t,i)$), we get the upper limit and the lower limit for $\Pi^*$.
If $\Pi^*$ is non-negative, we have
\begin{align}
 0 & =\sigma^2(t,i) \alpha e^{r(T-t)} \Pi^*  - (\mu(t,i) - r)  - K_1(t,i) \Theta_1(t) e^{-\alpha \Pi^* K_1(t,i) e^{r(T-t)}} \\
 & \qquad \qquad + K_2(t,i) \Theta_2(t) e^{\alpha \Pi^* K_2(t,i) e^{r(T-t)}}\\
 & > \sigma^2(t,i) \alpha e^{r(T-t)} \Pi^*  - (\mu(t,i) - r) - K_1(t,i) \Theta_1(t) e^{-\alpha \Pi^* K_1(t,i) e^{r(T-t)}} \\
 & \ge \sigma^2(t,i) \alpha e^{r(T-t)} \Pi^*  - (\mu(t,i) - r) - K M_1 e^{-\alpha \Pi^* K_1(t,i) e^{r(T-t)}}\\
 & \ge \sigma^2(t,i) \alpha e^{r(T-t)} \Pi^*  - (\mu(t,i) - r) - K M_1,
\end{align}
which implies
\begin{equation}
  \Pi^*(t,i) \leq \frac{\mu(t,i) - r + KM_1}{\sigma^2(t,i)\alpha e^{r(T-t)}}, 
\end{equation}
for all $(t,i) \in [0,T] \times \mathcal X$.
Otherwise, if $\Pi^*$ is non-positive, we get 
\begin{align}
	0&=\sigma^2(t,i) \alpha e^{r(T-t)} \Pi^*  - (\mu(t,i) - r)  - K_1(t,i) \Theta_1(t) e^{-\alpha \Pi^* K_1(t,i) e^{r(T-t)}} \\
    & \qquad \qquad + K_2(t,i) \Theta_2(t) e^{\alpha \Pi^* K_2(t,i) e^{r(T-t)}}\\
	& < - (\mu(t,i) - r) + K_2(t,i) \Theta_2(t) e^{\alpha \Pi^* K_2(t,i) e^{r(T-t)}} \\
	& \le - (\mu(t,i) - r) + M_2 e^{\alpha \Pi^* e^{r(T-t)}},
\end{align}	that leads to 
\begin{equation}
\Pi^*(t,i) \ge \frac{\ln \left( \frac{\mu(t,i)-r}{M_2}\right)}{\alpha e^{r(T-t)}}, \end{equation}
for all $(t,i) \in [0,T] \times \mathcal X$.
\end{proof}

\subsubsection{The Verification Theorem}
Now, we are ready to state the verification result.
\begin{theorem}[Verification Theorem] \label{thver} 
	Suppose that the Cauchy problem \eqref{PCvarphi} admits a classical solution $\varphi(\cdot,i) \in C^{1}\big{(} (0,T[\big{)} \cap C\big{(} [0,T] \big{)}$, for each $i \in \mathcal{X}$. Then, the function $\bar{V}:[0,T] \times \R \times \mathcal{X} \longrightarrow \R$ defined by 
	\begin{equation}
	\label{barV}
	\bar{V}(t,w,i)=-e^{-w \alpha e^{r(T-t)}} \varphi(t,i)
	\end{equation} 
	is the value function in \eqref{Vbaruesp}. Consequently, the strategy $\Pi_t^*=\Pi^*(t,X_t)$ described in Proposition \ref{opwi} is an optimal control.
\end{theorem}
\noindent The proof can be found in Appendix \ref{app:tech}.
\begin{remark}
By Theorem \ref{thver}, the value function \eqref{Vbaruesp} can be characterized as a transformation of the solution $\varphi$ to a certain system of ODEs with a particular terminal condition. As regards existence and uniqueness of a solution to this specific Cauchy problem \eqref{PCvarphi}, we refer to \citet[Theorem VII, Chapter II:6]{WALTER2} or to \citet[Section $6$]{BARAN}. 
According to \citet{WALTER2}, if $H$ given in \eqref{eq:H} is a locally Lipschitz function with respect to the second variable, uniformly in $t$, we get that there exists a unique solution $\varphi(t,i)$, for every $t \in [0,T]$, for all $i \in \mathcal X$. Requiring that $\mu$, $\sigma$, $K_1$ and $K_2$ are continuous functions is a sufficient condition for the regularity of function $H$ and, as a consequence, the smoothness of $\varphi$.
Otherwise, \eqref{PCvarphi} can be seen as a trivial case of the Cauchy problem faced by \citet{BARAN}. Assuming that $\mu(\cdot,i)$ and $\sigma(\cdot,i)$ are continuous functions in $t \in [0,T]$, for all $i \in \mathcal{X}$, guarantees that $\inf_{\Pi \in \R} \bar{\Psi}(t,i)$ is bounded with respect to the first variable and thus all required hypotheses are satisfied.
\end{remark}

%However, for reader's convenience all details can be found in Appendix \ref{app:tech}. 
\noindent The next result presents the optimal investment strategy for Problem \ref{prbesp}.
\begin{proposition}
Assume existence and uniqueness of a classical solution to the HJB equation with final condition \eqref{VbarHJB}. Moreover, suppose that for all $(t,i) \in [0,T] \times \mathcal X$,
\begin{equation}\label{sigma_IP}
    \sigma(t,i) > \sigma > 0.
\end{equation}
Then, the process $\{\Pi^*(t,i),\ t \in [0,T]\}$ characterized in Proposition \ref{opwi} provides the optimal investment strategy for Problem \ref{prbesp}.
\end{proposition}

\begin{proof}
Let
\begin{equation}
    \eta(t,i)=\max{ \left\{  \frac{\left|\ln \left( \frac{\mu(t,i)-r}{M_2}\right) \right|}{\alpha e^{r(T-t)}} ,\frac{\mu(t,i) - r + CM_1}{\sigma^2(t,i)\alpha e^{r(T-t)}}\right\} } , \quad (t,i) \in [0,T] \times \mathcal X.
\end{equation}
We show that conditions \eqref{eq:pi-eta} and \eqref{eta_cond} in Proposition \ref{suff} are satisfied. By Proposition \ref{properties}, we immediately have $\Pi^*(t,X_t) \leq \eta(t,X_t)$ and $\Pi^*(t,X_t) \geq -\eta(t,X_t)$, for every $t \in [0,T]$. Moreover, by \eqref{sigma_IP} and Assumption \ref{int_condit}, we get condition \eqref{eta_cond}. Then, the process $\{\Pi^*(t,i),\ t \in [0,T]\}$ is an admissible investment strategy and the statement follows by applying the Verification Theorem \ref{thver} and Proposition \ref{opwi}.
\end{proof}
%}

\subsection{The investment problem with the insurance derivative}

Now, we suppose that the insurer can write a pure endowment contract, whose payoff is given in \eqref{payoffPE}.

\noindent The following result ensures that the financial-insurance model outlined in Section \ref{sec:model} has a Markovian structure, i.e. the vector process $(W^\Pi,\Lambda,X)$ is a $(\mathbb G,\P)$-Markov-process.  
Let $\mathcal L_i^\Pi$ denote the Markov generator of $(W^\Pi,\Lambda,X)$ associated with a constant control $\Pi \in \R$. 
\begin{definition}
    The set $\mathcal D(\mathcal L_i^\Pi)$ denotes the class of functions $f (\cdot,\cdot,\cdot,i) \in C^1([0,T]) \times C^2(\R \times \R^+)$, for each $i \in \mathcal X$, such that for every constant $\Pi \in \R$, we have
    \begin{align} \label{ip1mkv}
     \mathbb{E} \bigg[ \int_{0}^{T} \Big( \sigma(v,X_{v})\Pi \pd{f}{w}(v,W_{v}^\Pi,\lambda_{v},X_{v}) \Big)^{2} \ud v \bigg] < \infty,\\
\mathbb{E} \bigg[ \int_{0}^{T} \Big(  c(v,\lambda_{v})\lambda_v\pd{f}{\lambda}(v,W_{v}^\Pi,\lambda_{v},X_{v}) \Big)^{2} \ud v \bigg] < \infty,  
    \end{align}
    and 
    \begin{align}\label{ip2mkv}
       & \esp{\int_0^T \int_{\R}  \left|f\big{(}v,W_{v}^\Pi,\lambda_{v},X_{v-\!}+h(X_{v-\!},z) \big{)} - f(v,W_{v}^\Pi,\lambda_{v},X_{v-\!})\right| m(\ud z) \ud v} < \infty,\\
        &\esp{\int_0^T \left|f\big{(}v,W_{v-}^\Pi+ \Pi K_1(v,X_{v-}),\lambda_{v},X_{v-}) \big{)} - f(v,W_{v-}^\Pi,\lambda_{v},X_{v-})\right|  \Theta_1(v) \ud v } < \infty,\\
        &\esp{\int_0^T \left|f\big{(}v,W_{v-}^\Pi- \Pi K_2(v,X_{v-}),\lambda_{v},X_{v-}) \big{)} - f(v,W_{v-}^\Pi,\lambda_{v},X_{v-})\right|  \Theta_2(v) \ud v} < \infty.
    \end{align}
\end{definition}
\begin{lemma}
	\label{wysx}
The stochastic process $(W^\Pi,\Lambda,X)$ is a Markov process on $(\Omega, \mathcal{F},\P;\mathbb{G})$, with infinitesimal generator $\mathcal L_i^\Pi$ for all constant strategies $\Pi \in \R$ given by   
	\begin{equation}
	\label{margen4}
	\begin{split}
	& \mathcal{L}_i^\Pi f (t,w,\lambda,i) = \pd{f}{t}(t,w,\lambda,i) + \big[ rw + (\mu(t,i)-r)\Pi \big]\pd{f}{w}(t,w,\lambda,i) \\
    & \quad + b(t,\lambda) \lambda \pd{f}{\lambda}(t,w,\lambda,i) + \frac{1}{2} \sigma^2(t,i)\Pi^{2} \pds{f}{w}(t,w,\lambda,i) + \frac{1}{2}c^{2}(t,\lambda)\lambda^2\pds{f}{\lambda}(t,w,\lambda,i) \\
    & \quad + \sum_{j \in \mathcal{X}}\! a_{ij} f(t,w,\lambda,j) + \Theta_1(t) \big{\{} V(t,w+\Pi K_1(t,i),\lambda,i) -V(t,w,\lambda,i) \big{\}} \\
    & \quad + \Theta_2(t) \big{\{} V(t,w-\Pi K_2(t,i),\lambda,i) -V(t,w,\lambda,i) \big{\}},
	\end{split}
	\end{equation}
for every $i \in \mathcal X$. The domain of the generator $\mathcal L_i^\Pi$ is $\mathcal D(\mathcal L_i^\Pi)$, for each $i \in \mathcal X$.
\end{lemma}
\noindent The proof is postponed to Appendix \ref{app:tech}. 

\noindent Let us consider the HJB equation that the value function $V$ is expected to solve, if sufficiently smooth:
\begin{equation}\label{VHJB}
\left \{ 
\begin{array}{ll} 
& \sup_{\Pi \in \R} {\mathcal L}_i^\Pi V(t,w,\lambda,i) + \lambda \big( \bar{V}(t,w,i)- V(t,w,\lambda,i) \big) = 0, \\ %= -\sum_{j \in \mathcal{X}}\! a_{ij} V(t,w,y,s,j), \\
& \qquad \qquad \qquad \forall (t,w,\lambda,i) \in [0,T) \times \R \times \Rp \times \mathcal X,\\
& V(T,w,\lambda,i)=-e^{-\alpha(w-K)} \\
& \qquad \qquad  \qquad \forall (w,\lambda,i) \in \R \times \Rp \times \mathcal X,
\end{array} 
\right.
\end{equation}
      How to derive the HJB equation in \eqref{VHJB} is shown in Appendix \ref{app:HJB}.
      
\noindent Now, 
let us introduce the following ansatz 
\begin{equation}\label{V}
	V(t,w,\lambda,i)=-e^{-w \alpha e^{r(T-t)}} \varphi(t,i) \phi(t,\lambda),
\end{equation}
with $(t,w,\lambda,i) \in [0,T] \times \R \times \Rp \times \mathcal{X}$, where $\varphi$ solves \eqref{PCvarphi}, while the function $\phi$ is non-negative and does not depend on $w$. 

\noindent From \eqref{V}, replacing all the derivatives and performing some computations, 
problem \eqref{VHJB} reduces to 
\begin{equation}
\label{PCphi}
\left \{\begin{array}{ll}
\pd{\phi}{t}(t,\lambda)+b(t,\lambda)\lambda \pd{\phi}{\lambda}(t,\lambda) + \ds \frac{1}{2}c^{2}(t,\lambda)\lambda^2 \pds{\phi}{\lambda}(t,\lambda) - \lambda (\phi(t,\lambda) -1) = 0, \\
\quad  \forall (t,\lambda) \in [0,T) \times \R^+,  \\
\phi(T,\lambda)=e^{\alpha K},\quad \forall \lambda \in \R^+.
\end{array} \right.
\end{equation}	

\noindent We observe that the PDE in \eqref{PCphi} is linear and a solution exists under suitable conditions on model coefficients; see, e.g. \citet[Theorem 5.3]{pham1998optimal} or \citet[Theorem 1]{colaneri2019classical}.

Clearly, if the function $\phi$ is a classical solution of the Cauchy problem
\eqref{PCphi},
then $V(\cdot,\cdot,\cdot,i) \in C^{1,2,2}([0,T] \times \R \times \Rp)$, for each $i \in \mathcal{X}$ and we have that $V(t,w,\lambda,i)=-e^{-w \alpha e^{r(T-t)}} \varphi(t,i) \phi(t,\lambda)$ solves the original HJB equation given in \eqref{VHJB}. Now, we can state the verification result. 
\begin{theorem}[Verification Theorem] \label{thver2} 
	%Suppose that the Cauchy problems \eqref{PCvarphi} and \eqref{PCphi} admit bounded classical solutions
	Let $\varphi(\cdot,i) \in C^{1}\big{(} (0,T)\big{)} \cap C\big{(} [0,T] \big{)}$ and $\phi(\cdot,\cdot) \in C^{1}\big{(} (0,T) \times \Rp \big{)} \cap C\big{(} [0,T] \times \Rp\big{)}$, for each $i \in \mathcal{X}$, be classical solutions of the Cauchy problems \eqref{PCvarphi} and \eqref{PCphi}, respectively. 
Then, the function $V:[0,T] \times \R \times \Rp \times \mathcal{X} \longrightarrow \R$ defined by \eqref{V}
is the value function in \eqref{Vuesp}.
Consequently, the strategy $\Pi_t^*=\Pi^*(t,X_t)$ described in Proposition \ref{opwi} is an optimal control.
\end{theorem}
\noindent The proof can be found in Appendix \ref{app:tech}.
\begin{remark}
We observe that the optimal investment strategy $\Pi^*(t,X_t)$ turns out to be the same as the pure investment problem. 
This means that the optimal portfolio for the investment problem with the insurance derivative is equal to the strategy without insurance risks when the insurance payment is independent of the risky asset price process. This statement is the same as that given in \citet{Delong2009} and \citet{LiangLu2017} when the risky asset price dynamics is driven by a Lévy process and a shot-noise process, respectively.
\end{remark}

\section{Characterization of the indifference price}\label{sec:indiff-price}

In this section we compute explicitly the indifference price $P$, see Definition \ref{indiff-price}, for the pure endowment contract whose payoff is given in \eqref{payoffPE}.
%To this aim, we provide the definition of the indifference price charged by an insurer which writes a pure endowment. 
%Recall that $\bar V$ and $V$ are the value functions introduced in \eqref{Vbaruesp} and \eqref{Vuesp},
%respectively.
%\begin{definition}
%Given $W_{t}=w$, $\lambda_{t}=\lambda$ and $X_{t}=i$, the {\em indifference price} process or {\em reservation price} process $P=\{P_t,\ t \in [0,T]\}$ of the insurer related to the pure endowment contract is defined at any time $t \in [0,T]$ as the $\mathbb{G}$-adapted process implicit solution to the equation 
%	\begin{equation}
%	\label{przz}
%		\bar{V}(t,w,i)=V(t,w+P_{t},\lambda,i).
	%\end{equation}
    %\end{definition} 
Recall that the indifference price is the value $P$ at which the insurer is indifferent, in terms of expected utility, between not selling the policy and selling it for $P$
now while committing to pay the benefits at maturity. \\
\noindent We have the following explicit characterization of the indifference price process in our framework.
\begin{proposition} 
\label{thprezzo}
Under the same hypotheses of Theorem \ref{thver2}, for every $t \in [0,T]$, the indifference price of the insurer related to the pure endowment with maturity $T$ is given by
\begin{equation}
\label{prezzo}
P_{t}=P(t,\lambda;T) = \frac{\ln \big( \phi(t,\lambda) \big) }{\alpha e^{r(T-t)}},
\end{equation} for all $(t,\lambda) \in [0,T] \times \Rp $, where the function $\phi$ solves the Cauchy problem \eqref{PCphi}.
\begin{proof}
By Theorem \ref{thver} and Theorem \ref{thver2}, equation \eqref{przz} reads as
\begin{equation*}
	e^{-w\alpha e^{r(T-t)}} \varphi(t,i) = e^{-(w+P_{t}) \alpha e^{r(T-t)}}  \varphi(t,i) \phi(t,\lambda),
\end{equation*} 
and then \begin{equation*} e^{P_{t} \alpha e^{r(T-t)}}= \phi(t,\lambda),	\end{equation*} 
from which, computing the logarithm of both members, we get \eqref{prezzo}.
\end{proof}
\end{proposition}

\begin{remark}
	In the actuarial literature, %an exponential utility function %like $u(w)=-e^{-\alpha w}$, for $w \in \R$,
	exponential utility preferences 
	imply that the indifference price of a pure endowment depends on the risk-aversion coefficient, the stochastic interest rate and the logarithm of the function that links the two value functions and is independent of wealth (see e.g. \citet{YOUNGZAR, Young, YOUNGMOORE, LudkowskiYoung2008}). Here, the indifference price shares the same features. Moreover, we note that it does not explicitly depend on the current regime; 
meaning that in our model, the state of the economy does not affect directly the indifference price of such type of insurance contracts but only the amount invested in the financial market. %Finally, we also note that the reservation price is independent of wealth, as in any model where the investor's preferences are described by an exponential utility function.
\end{remark}

\noindent Under the indifference pricing principle, the premium satisfies a suitable backward PDE.
\begin{corollary}
For every $(t,\lambda) \in [0,T] \times \Rp$ the indifference premium $P=P(t,\lambda;T)$ satisfies the following PDE \begin{equation}
\begin{split}
	rP & = \pd{P}{t}+b(t,\lambda)\lambda \pd{P}{\lambda} + \frac{1}{2}c^{2}(t,\lambda)\lambda^2 \Big{\{}  \pds{P}{\lambda} + \alpha e^{r(T-t)}\left( \pd{P}{\lambda}\right)^2 \Big{\}} \\
    & \qquad + \frac{\lambda}{\alpha e^{r(T-t)}} \bigg(e^{-P \alpha e{r(T-t)}} -1 \bigg),
    \end{split}
\end{equation} 
with boundary condition $P(T,\lambda;T)=K$, for each $\lambda \in \Rp$.
\begin{proof}
It follows from a straightforward application of  \eqref{PCphi} and \eqref{prezzo}.	
\end{proof}
\end{corollary}

\noindent Finally, we provide a probabilistic representation for the indifference price process $P$. %by means of the FeynmanâKac formula. 
Indeed, if the function $\phi$ solves the Cauchy problem \eqref{PCphi}, we can represent $\phi$ as an expectation via an extension of the Feynman-Kac formula. More precisely, using the linear PDE for $\phi -1$, it is easy to see that 
\begin{equation}
\phi(t,\lambda)-1 = \mathbb{E}_{t,\lambda} \big[ e^{-\int_{t}^{T}\lambda_{v} \ud v} \big(e^{\alpha K} - 1 \big) \big],
\end{equation} for every $(t,\lambda) \in [0,T] \times \Rp$. So, as a consequence, we have that 
\begin{equation}
\phi(t,\lambda)= 1 + \big( e^{\alpha K} - 1 \big)\mathbb{E}_{t,\lambda} \big[ e^{-\int_{t}^{T}\lambda_{v} \ud v} \big],
\end{equation} where $ \mathbb{E} _{t,\lambda} $ denotes the conditional expectation given $\lambda_{t}=\lambda$, for every $(t,\lambda) \in [0,T] \times \Rp$. We outline that $\mathbb{E}_{t,\lambda} \Big[ e^{-\int_{t}^{T}\lambda_{v} \ud v} \Big]$ is the conditional probability that an individual will survive until time $T$ given that she/he is alive at time $t$. Hence, representing the function $\phi$ as 
\begin{equation}
\label{fkphi}
\phi(t,\lambda)= e^{\alpha K} \mathbb{E}_{t,\lambda} \big[ e^{-\int_{t}^{T}\lambda_{v} \ud v} \big] + \Big(1 - \mathbb{E}_{t,\lambda} \big[ e^{-\int_{t}^{T}\lambda_{v} \ud v} \big] \Big) = \mathbb{E}_{t,\lambda} \big[ e^{\alpha G_{T}} \big] , \end{equation}
for every $(t,\lambda) \in [0,T] \times \Rp$, the indifference price of the insurer related to a pure endowment contract can be written as
\begin{equation}
	P_{t}=P(t,\lambda;T) = \frac{\ln \Big( \mathbb{E}_{t,\lambda} \big{[} e^{\alpha G_{T}} \big{]} \Big) }{\alpha e^{r(T-t)}},
\end{equation} for every $(t,\lambda) \in [0,T] \times \Rp$.

\section{Some numerical examples}\label{sec:numerical}

In this section, we provide numerical results to illustrate the qualitative behavior of the proposed model, highlighting features that are challenging to analyze solely through theoretical methods. Specifically, we aim to examine the impact of regime-switching dynamics and the stochastic hazard rate on the insurer's optimal strategies and pricing decisions. These simulations offer valuable insights into the interplay between economic regimes and mortality risks, demonstrating how these factors influence the insurer's risk management and indifference pricing outcomes.

To simplify the analysis, we suppose that the Markov chain $X$ has two states, namely $\mathcal{X}=\{1,2\}$, that can be interpreted as the 'Good' and 'Bad' economic regimes, respectively. For instance, the good regime could represent a market in economic boom whereas the bad regime could be a market in economic recession in which security prices are expected to fall. We also call these two regimes of the market 'bull' market and 'bear' market, respectively.

First of all, we have to set values for the infinitesimal generator of the $2$-state Markov chain. Since $a_{ij}$ represents the average of number of switches in an unit time, from state $i$ to $j$, and since empirical observations of the market suggest that it is more likely to pass from a good economic state to a bad one than the opposite, we choose $a_{12}>a_{21}$. In particular, we take $a_{12}=0.2$ and $a_{21}=0.1$.

\noindent For the sake of simplicity, we assume that functions $\mu$, $\sigma$, $K_1$ and $K_2$ depend only on the Markov chain. So, by \eqref{s}, the risky asset price dynamics is given by
\begin{equation*}
	\ud S_{t}=S_{t-}\left(\mu_{i}\ud t + \sigma_{i} \ud Z^{S}_{t} + K_{1,i} \ud N_t^1 - K_{2,i} \ud N_t^2\right), \quad  S_{0}>0, \quad i=1,2,
\end{equation*} 
where $\mu_{i}$, $\sigma_{i}$, $K_{1,i}$ and $K_{2,i}$ denote the expected rate of return, the volatility and the jump coefficients in the $i$-th regime.
By way of example, we set the initial value of the stock price to be $S_{0}=1$ and the short-term interest rate to be $r=5\%$.  
As shown by \citet{FRENCH}, the appreciation rate of the underlying risky asset 
%and the interest rate are 
is higher in a growing economy, so we assume that $\mu_{1}>\mu_{2}$. Moreover, in each economic regime, the return of the risky asset should be higher than that of the risk-free rate, as required also in our modeling framework. \citet{HAMLIN} find that economic recessions represent the main factor that drives fluctuations in the volatility of stock returns, so we assume that volatility is lower in a good economy, i.e. $\sigma_{1}<\sigma_{2}$. Furthermore, let us assume that $\ds \frac{\mu_{1}-r}{\sigma_{1}^{2}}>\frac{\mu_{2}-r}{\sigma_{2}^{2}}$. In fact, according to \citet{FRENCH}, even though the expected market risk premium (defined as the expected return on the stock minus the risk-free interest rate) is usually higher during a 'bear' market than during a 'bull' market, the volatility of the stock offsets the effect of this quantity and, as a consequence, the ratio 'expected excess return/return variance' is greater when the economic conditions are good. 
As for the jump terms, we consider two homogeneous Poisson processes $N^1$ and $N^2$ with constant intensities $\Theta_1=0.3$ and $\Theta_2=0.4$. 
We observe that the higher are the values of function $K_1$, the higher is the price of the stock $S$. On the other hand, any increase in the coefficient $K_2$ leads to smaller prices for the risky stock. Moreover, we notice that large values of $K_2$ cause dizzying upward or downward peaks for the stock price, even though the intensity $\Theta_2$ is tiny. Therefore, since in a market with good economic conditions stock prices are rising or are expected to rise, we suppose that $K_{1,1}>K_{1,2}$ and $K_{2,1}<K_{2,2}$.  
In particular, we choose the parameters as in Table \ref{tab}.
    \begin{table}[H]
	\caption{Simulation market parameters.}
	\label{tab}
    \centering
	\begin{tabular}{ccccc}
		\toprule
		Regime  & $\mu$ & $\sigma$ & $K_1$ & $K_2$ \\
		\midrule
		'Good'  & $0.15$ & $0.15$ & $0.15$ & $0.3$ \\
		'Bad'  & $0.12$ & $0.25$ & $0.1$ & $0.35$ \\
		\bottomrule
\end{tabular}\end{table} 

\noindent Since the underlying market is a continuous-time model, we need to discretize it by Monte Carlo simulation. 
The time horizon is taken to be $T=10$ years and we discretize time with a total of $1000$ time steps (that means that we take into account about two updates of $S$ every workweek), each of width $\Delta t=\frac{1}{100}$. 
To get an idea of our model, we simulate three trajectories of the risky asset $S$ in Figure \ref{S_MC}. We notice that the stock price is greater during a 'bull market' rather than during a 'bear' market and it also exhibits jumps at switching times of the Markov chain.
\begin{figure}[h!]
	\centering
	\includegraphics[width=7.5cm]{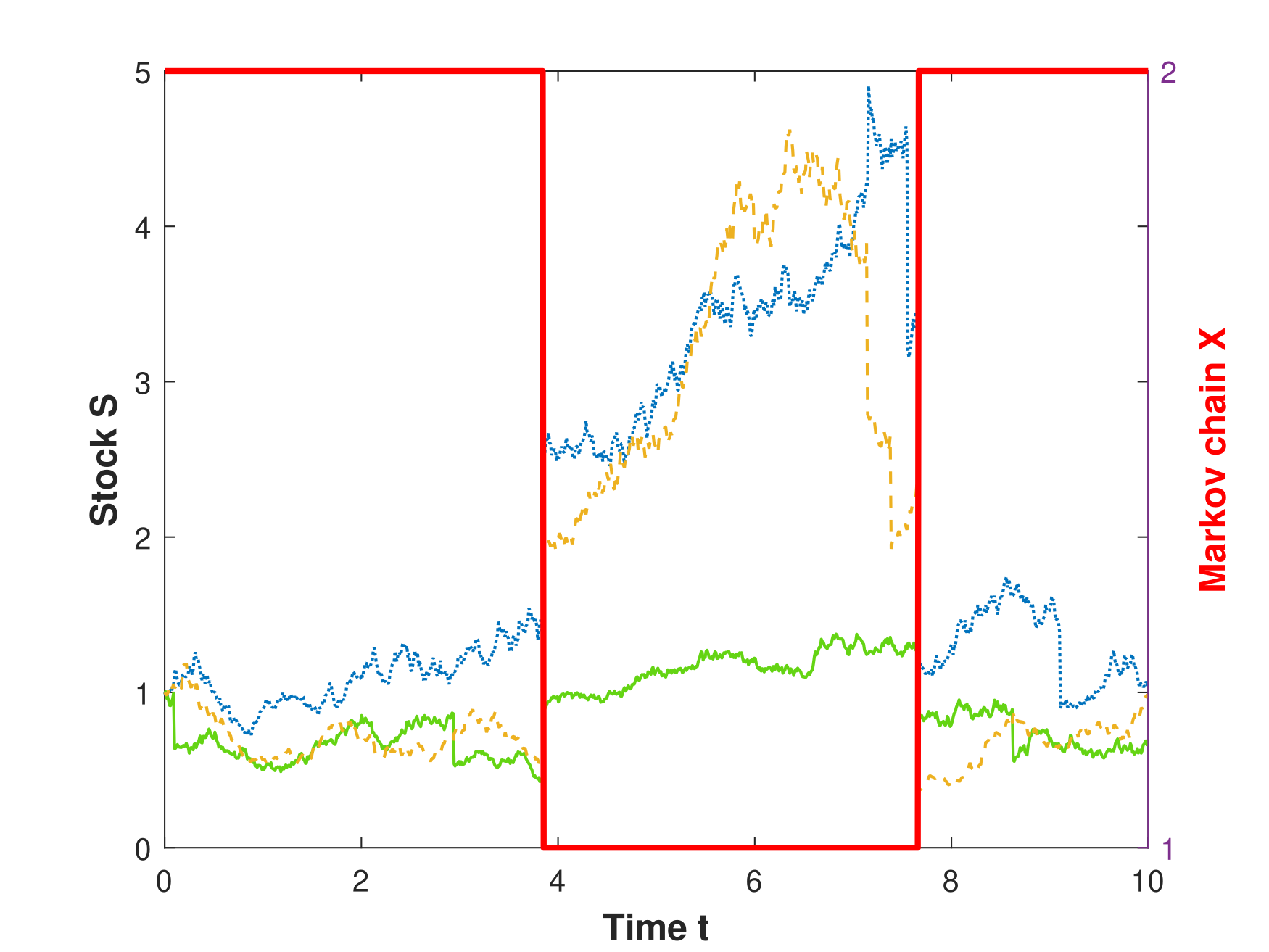}
	\caption{The effect of the regime-switching on the stock price $S$.}
	\label{S_MC}	
\end{figure}

\noindent We set the risk aversion coefficient to $\alpha= 0.5$.
Next, we compute the optimal investment strategy based on Proposition \ref{opwi}, in order to investigate how it is sensitive to economic regimes. In Figure \ref{optstr(t)_MC} we plot the optimal dynamic portfolio given by \eqref{eqh}, as a function of time. 
\begin{figure}[h!]
	\centering
	\includegraphics[width=7.5cm]{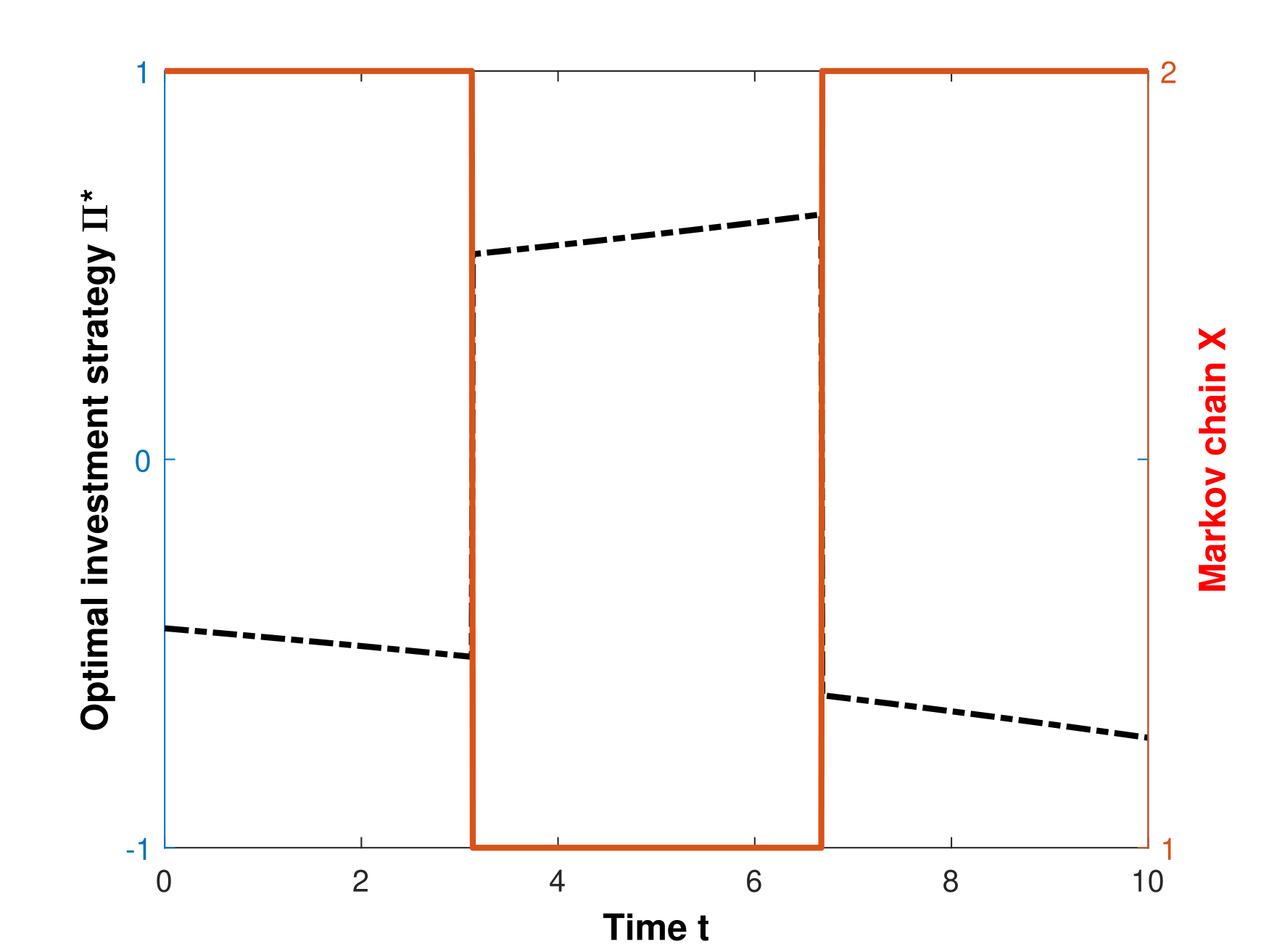}
	\caption{The effect of regime-switching on the optimal strategy $\Pi^*$.}
	\label{optstr(t)_MC}	
\end{figure} 
\noindent We clearly notice that a regime switch leads to a sudden change in the optimal strategy.
Moreover, we note that in a good economy the amount invested in the stock is always positive and increasing with respect to time; instead, if the market scenario is bad, the strategy is negative; this indicates that when the economic conditions are bad, the insurer prefers to short-sell the risky asset. 
Next, we consider a life insurance policy and examine its indifference price.

\noindent In the current toy example of our proposed model, we assume that the hazard rate follows a mean-reverting Brownian Gompertz model, similar to the one proposed in \citet{MILPRO}, i.e. 
\begin{align} 
	\lambda_{t} & = \lambda_{0} e^{c_1t+c_2 Y_{t}}, \quad c_1, c_2, \lambda_{0}>0, \\ \ud Y_{t} & =-mY_{t}\ud t+\ud Z_{t}^{Y}, \ \ Y_{0}=0, \quad  m\geq 0,
\end{align} 
with $c_1=0.083$, $c_2=0.1$, $\lambda_{0}=0.01$ and $m=0.5$. Let us observe that this choice corresponds to \eqref{y}, considering $b(t,\lambda)=c_1+m \ln(\lambda_{0}) +\frac{1}{2}c_2^2-m \ln(\lambda) + mc_1t$ and $c(t,\lambda)=c_2\lambda$, for all $(t,y)$. This model guarantees that the hazard rate is kept positive and does not explode on $[0,T]$, as it is an exponential function that depends on a stochastic factor $Y$ with a mean reversion behavior.

In this context, based on the results obtained above, we compute the indifference price of an insurer related to a pure endowment contract that pays $K=1$ if the policyholder is still alive after $10$ years of acquiring the policy. 
Thus, the payoff is easily given by the random variable
$$
G_T= \I_{\{\tau > T \}},
$$ 
recalling that $\tau$ represents the remaining lifetime of the insured.
Now, we aim to analyze the value function $\bar{V}$ associated with the insurer who solely invests her/his wealth in the market, and the value function $V$ corresponding to the insurer who also underwrites a pure endowment contract.
\begin{figure}[h!]
	\centering
	\includegraphics[width=10cm]{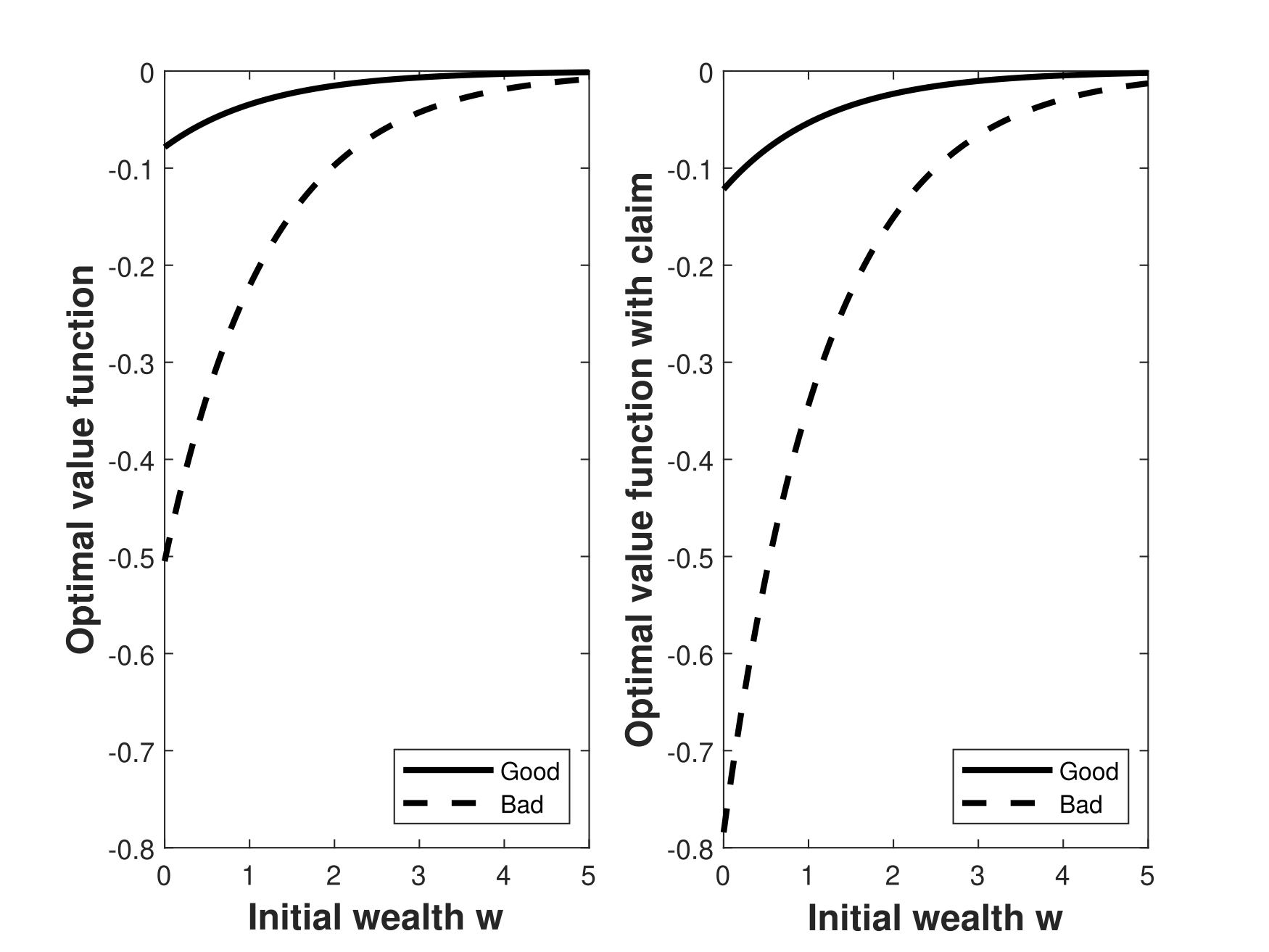}
	\caption{Optimal value at time $0$ as a function of wealth when the economic regime is $i=1$ (solid line) or $i=2$ (dashed line). Left panel: the pure investment problem. Right panel: the investment problem with the insurance contract.}
	\label{V(w)}	
\end{figure}
\noindent Figure \ref{V(w)} depicts the value functions $\bar{V}$ (left panel) and $V$ (right panel) at time $t=0$, with respect to the initial wealth $w$, associated to the optimal strategy computed above, when the market state is good (solid line) or bad (dashed line). The two panels exhibit the same behavior: the optimal value functions are increasing functions of wealth, in both regimes. It is worth noting that values reached by functions $\bar{V}$ and $V$ are always higher in a 'bull' market, as it is reasonable. Furthermore, we can also point out that different economic conditions imply different value functions and that this gap becomes greater when the insurer, beyond investing in the financial market, also writes an insurance contract. 

\noindent Next, we investigate the indifference price of a pure endowment policy, in order to highlight the dependence of a life insurance contract on mortality force and time to expiration.
In view of the probabilistic representation provided in Section $5$, the indifference price charged by the insurer is determined as \begin{equation}\label{PREZZOESNUM}
	P_{t}=P(t,\lambda;T) = \frac{\ln \Big( 1 + (e^{\alpha}-1)\mathbb{E}_{t,\lambda} \big[ e^{-\int_{t}^{T}\lambda_{v} \ud v} \big] \Big) }{\alpha e^{r(T-t)}},
\end{equation} for every $(t,\lambda) \in [0,T] \times \Rp$.
We employ this formula, using the standard Monte Carlo method (with parameter $M=5000$) to evaluate expectations with respect to the probability measure $\P$.
From expression \eqref{PREZZOESNUM}, we point out that economic regimes do not affect the price which instead strongly depends on the risk aversion coefficient and the risk-free interest rate. In particular, it is easy to see that the indifference price increases as risk aversion increases and, at the same time, it decreases as long as the interest rate increases. Furthermore, since the dependence on the mortality rate $\lambda$ is not explicit, we would like to numerically analyze how the hazard rate affects the price of the life insurance policy involved. 
First of all, we show the impact of changing initial mortality rate on the indifference price charged at the beginning of the time interval.

In Figure \ref{P(lambda0)} we observe that the larger force of mortality decreases the indifference price $P_0$ charged by the insurer at time $t=0$, as is reasonable to expect for such types of insurance contracts. This is consistent with common intuition
as, under higher mortality, an endowment payout is less likely.
\begin{figure}[h!]
	\centering
	\includegraphics[width=7.5cm]{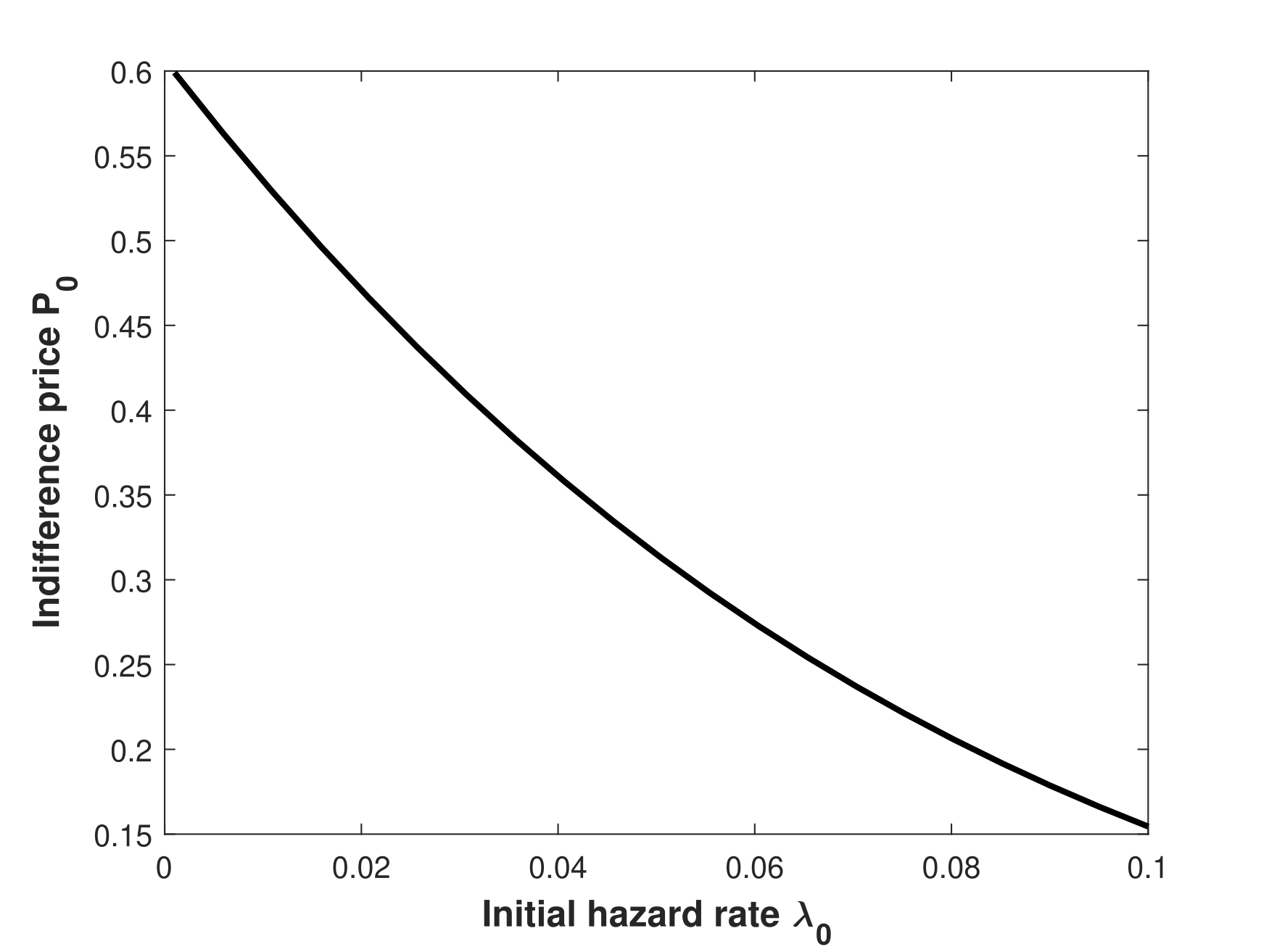}
	\caption{The effect of the hazard rate on the indifference price at time $t=0$.}
	\label{P(lambda0)}	
\end{figure}

\noindent Finally, we investigate the evolution of indifference price over time. For simplicity, we assume constant mortality (such as in some numerical experiments of \citet{YOUNGMOORE}). In this framework, we calculate the indifference premium related to a pure endowment policy for the insurer and we plot it as a function of time to maturity. 
\begin{figure}[t]
	\centering
\includegraphics[width=7.5cm]{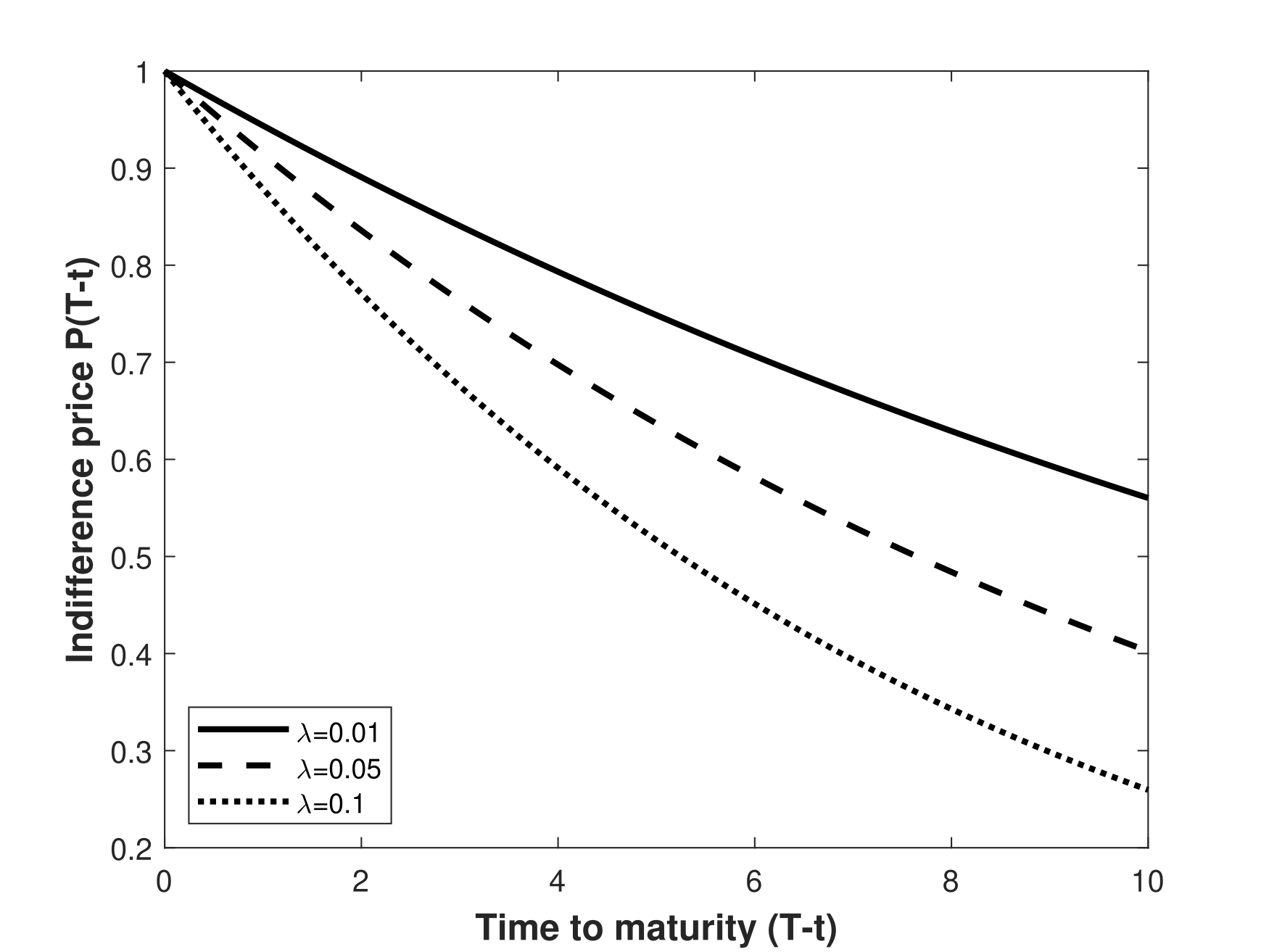}
	\caption{The effect of the hazard rate on the indifference price for several different deferral periods.}
	\label{P(T-t)}
\end{figure}
As before, in Figure \ref{P(T-t)}, we can see that the higher the hazard rate, the lower the indifference price for a pure endowment policy, whether the economic conditions are good or not; in other terms, the price is more sensitive to variations of deferral periods, when the intensity of population mortality is more pronounced. Moreover, it is worth noting that the indifference price is a decreasing function of time of maturity, i.e. the premium is bigger. %for shorter deferral periods
as time approaches maturity, as usually happens.

\section{Concluding remarks}\label{sec:conclusion}

In this paper, we have analyzed the indifference pricing of pure endowments within a regime-switching model for an insurer with exponential utility preferences. The financial market model considered consists of a riskless asset and a stock whose price evolves according to a jump diffusion process modulated by a continuous-time finite-state Markov chain, representing the states of the economy. The model accounts for both smooth price movements and abrupt changes due to positive and negative jumps, which reflect favorable events (e.g., market expansions) and adverse shocks (e.g., regulatory changes), respectively. In addition, we have incorporated a stochastic hazard rate to describe population mortality.\\
Within this framework, using the actuarial principle of equivalent utility, we have characterized the indifference price for a pure endowment contract and derived its probabilistic representation. Furthermore, we have demonstrated that the indifference price solves a final value problem. Specifically, the premium that makes the insurer indifferent—based on expected utility between not selling the policy and selling it for that premium while paying the benefits at maturity is linked to a classical solution of a linear PDE with a proper terminal condition. This was achieved by solving an equation involving two value functions corresponding to the stochastic control problems with and without insurance liabilities. Applying the classical control approach based on the HJB equation, we have identified the optimal investment strategies and verified the value functions for both problems via solutions to a linear PDE and a system of ODEs.\\
A sensitivity analysis, conducted under a two-state Markov chain, has highlighted several key features of the indifference price. We have investigated the influence of the hazard rate and the time to maturity on the price. Our findings suggest that when the mortality intensity is low, the likelihood of an endowment payout increases, resulting in a higher premium. Furthermore, we have observed that the indifference price decreases for longer deferral periods, as expected.\\
Another noteworthy numerical result shows that the insurer chooses to short-sell the risky asset when the financial market is in the bad state (characterized by a low rate of return, high volatility, and frequent jumps in the stock price).\\
Building on the methodology developed here, future work could consider more complex life‐insurance products—such as equity-linked policies—and extend the mortality model by incorporating multiple hazard rates to account for heterogeneity (for example, age cohorts). Moreover, allowing the risk‐aversion parameter $\alpha$ to depend on the regime would capture "flight‐to‐safety" behavior more faithfully, namely the tendency of investors to become more risk-averse in turbulent markets.
%but would break the key separation of the wealth variable that underpins our linear PDE/ODE reduction, leading instead to a nonlinear interlinked system of PDEs. that 
However, this choice would result in a nonlinear HJB equation, which no longer permits the separation of the wealth variable and leads, instead, to a nonlinear fully coupled system of PDEs for which a classical solution is not guaranteed in general.
We leave a full treatment of these extensions to forthcoming research.

%Building on the methodology developed in this paper, future work could explore the evaluation of more complex insurance products, such as equity-linked policies. 
%Additionally, we could generalize the framework by incorporating multiple hazard rates to account for different characteristics of the insured individuals, with age being a primary factor.
%This extension would also consider alternative utility functions, such as power or logarithmic utility, to better capture the insurer's varying risk preferences. 

\medskip

\section*{Acknowledgements}

The first author is member of the Gruppo Nazionale  per l'Analisi Matematica, la Probabilità e le loro Applicazioni (GNAMPA) of the Istituto Nazionale di Alta Matematica (INdAM).

%\section*{Disclosure statement}

%No potential conflict of interest was reported by the authors.

%\section*{ORCID}

%\noindent {\em Alessandra Cretarola}\orcidlink{http://orcid.org/0000-0003-1324-9342} \url{http://orcid.org/0000-0003-1324-9342}\\
%\noindent {\em Benedetta Salterini}\orcidlink{http://orcid.org/0000-0001-6950-4766} \url{http://orcid.org/0000-0001-6950-4766}

\bibliographystyle{plainnat}
\bibliography{biblio}

\begin{thebibliography}{31}
\providecommand{\natexlab}[1]{#1}
\providecommand{\url}[1]{\texttt{#1}}
\expandafter\ifx\csname urlstyle\endcsname\relax
  \providecommand{\doi}[1]{doi: #1}\else
  \providecommand{\doi}{doi: \begingroup \urlstyle{rm}\Url}\fi

\bibitem[Baran et~al.(2013)Baran, Yin, and Zhu]{BARAN}
N.A. Baran, G.~Yin, and C.~Zhu.
\newblock Feynman-kac formula for switching diffusions: connections of systems of partial differential equations and stochastic differential equations.
\newblock \emph{Advances in Difference Equations}, 2013, 2013.
\newblock \doi{10.1186/1687-1847-2013-315}.

\bibitem[Basak et~al.(2011)Basak, Ghosh, and Goswami]{GOPAL}
G.K. Basak, M.K. Ghosh, and A.~Goswami.
\newblock Risk minimizing option pricing for a class of exotic options in a {M}arkov-modulated market.
\newblock \emph{Stochastic Analysis and Applications}, 29\penalty0 (2):\penalty0 259--281, 2011.
\newblock \doi{10.1080/07362994.2011.548665}.

\bibitem[Biffis(2005)]{biffis}
E.~Biffis.
\newblock Affine processes for dynamic mortality and actuarial valuations.
\newblock \emph{Insurance: Mathematics and Economics}, 37\penalty0 (3):\penalty0 443--468, 2005.
\newblock \doi{10.1016/j.insmatheco.2005.05.003}.

\bibitem[Brémaud(1981)]{BREMAUD}
P.~Brémaud.
\newblock \emph{Point processes and queues}.
\newblock Springer Verlag, 1981.

\bibitem[Cairns et~al.(2006)Cairns, Blake, and Dowd]{cairns}
A.~Cairns, D.~Blake, and K.~Dowd.
\newblock Pricing death: frameworks for the valuation and securization of mortality risk.
\newblock \emph{ASTIN Bulletin}, 36:\penalty0 79--120, 2006.
\newblock \doi{10.1017/S0515036100014410}.

\bibitem[Ceci et~al.(2020)Ceci, Colaneri, and Cretarola]{CeciEtAl2020}
C.~Ceci, K.~Colaneri, and A.~Cretarola.
\newblock Indifference pricing of pure endowments via bsdes under partial information.
\newblock \emph{Scandinavian Actuarial Journal}, 2020\penalty0 (10):\penalty0 904--933, 2020.
\newblock \doi{10.1080/03461238.2020.1790030}.

\bibitem[Choi(2017)]{choi2017indifference}
J.~Choi.
\newblock Indifference pricing of a glwb option in variable annuities.
\newblock \emph{North American Actuarial Journal}, 21\penalty0 (2):\penalty0 281--296, 2017.
\newblock \doi{10.1080/10920277.2017.1283237}.

\bibitem[Colaneri and Frey(2021)]{colaneri2019classical}
K.~Colaneri and R.~Frey.
\newblock Classical solutions of the backward {P}{I}{D}{E} for a {M}arkov modulated marked point processes and applications to {C}{A}{T} bonds.
\newblock \emph{Insurance: Mathematics and Economics}, 101:\penalty0 498--507, 2021.
\newblock \doi{https://doi.org/10.1016/j.insmatheco.2021.09.003}.

\bibitem[Dacorogna and Cadena(2015)]{DC2015}
M.M. Dacorogna and M.~Cadena.
\newblock Exploring the dependence between mortality and market risks.
\newblock \emph{SCOR Papers}, 4:\penalty0 1–31, 2015.
\newblock \doi{10.1017/asb.2024.20}.

\bibitem[Dahl(2004)]{DAHL04}
M.~Dahl.
\newblock Stochastic mortality in life insurance: market reserves and mortality-linked insurance contracts.
\newblock \emph{Insurance: Mathematics and Economics}, 35:\penalty0 113--136, 2004.
\newblock \doi{10.1016/j.insmatheco.2004.05.003}.

\bibitem[Davis(1993)]{MMO}
M.H.A. Davis.
\newblock \emph{Markov Models and Optimization}, volume~49.
\newblock CRC Press, 1993.

\bibitem[Davis et~al.(1993)Davis, Panas, and Zariphopoulou]{DavisEtAl1993}
M.H.A. Davis, V.G. Panas, and T.~Zariphopoulou.
\newblock European option pricing with transaction costs.
\newblock \emph{SIAM Journal on Control and Optimization}, 31\penalty0 (2):\penalty0 470--493, 1993.
\newblock \doi{10.1137/0331022}.

\bibitem[Delong(2009)]{Delong2009}
L.~Delong.
\newblock Indifference pricing of a life insurance portfolio with systematic mortality risk in a market with an asset driven by a {L}\'{e}vy process.
\newblock \emph{Scandinavian Actuarial Journal}, 2009(1):\penalty0 1--26, 2009.
\newblock \doi{10.1080/03461230701795907}.

\bibitem[Eichler et~al.(2017)Eichler, Leobacher, and Sz{\"o}lgyenyi]{EichlerEtAl2017}
A.~Eichler, G.~Leobacher, and M.~Sz{\"o}lgyenyi.
\newblock Utility indifference pricing of insurance catastrophe derivatives.
\newblock \emph{European Actuarial Journal}, 7\penalty0 (2):\penalty0 515--534, 2017.
\newblock \doi{10.1007/s13385-017-0154-2}.

\bibitem[Elliott and Siu(2011)]{ElliottSiu2011}
R.J. Elliott and T.K. Siu.
\newblock Utility indifference pricing under regime switching models.
\newblock \emph{Journal of Computational and Applied Mathematics}, 235:\penalty0 1041--1056, 2011.

\bibitem[French et~al.(1987)French, Schwert, and Stambaugh]{FRENCH}
K.~French, G.~Schwert, and R.~Stambaugh.
\newblock Expected stock return and volatility.
\newblock \emph{Journal of Financial Economics}, 19\penalty0 (1):\penalty0 3--29, 1987.
\newblock ISSN 0304-405X.
\newblock \doi{10.1016/0304-405X(87)90026-2}.

\bibitem[Gihman and Skorohod(1972)]{GIHMSKOR}
I.I. Gihman and A.V. Skorohod.
\newblock \emph{Stochastic Differential Equations}.
\newblock Springer-Verlag, 1972.

\bibitem[Gyulov and Koleva(2022)]{gyulov2022penalty}
T.~B. Gyulov and M.~N. Koleva.
\newblock Penalty method for indifference pricing of american option in a liquidity switching market.
\newblock \emph{Applied Numerical Mathematics}, 172:\penalty0 525--545, 2022.
\newblock \doi{10.1016/j.apnum.2021.11.002}.

\bibitem[Hamilton and Gang(1996)]{HAMLIN}
J.~Hamilton and L.~Gang.
\newblock Stock market volatility and the business cycle.
\newblock \emph{Journal of Applied Econometrics}, 11\penalty0 (5):\penalty0 573--593, 1996.
\newblock \doi{10.1002/(SICI)1099-1255(199609)11:5<573::AID-JAE413>3.0.CO;2-T}.

\bibitem[Henderson and Hobson(2009)]{HendersonHobson2009}
V.~Henderson and D.~Hobson.
\newblock Utility indifference pricing: An overview.
\newblock In R.~Carmona, editor, \emph{Indifference pricing: Theory and Applications}, chapter~2, pages 44--73. Princeton University Press, 2009.

\bibitem[Hodges and Neuberger(1989)]{HodgesNeuberger1989}
S.D. Hodges and A.~Neuberger.
\newblock Optimal replication of contingent claims under transaction costs.
\newblock \emph{Review of Futures Markets}, 8\penalty0 (2):\penalty0 222--239, 1989.

\bibitem[Liang and Lu(2017)]{LiangLu2017}
X.~Liang and Y.~Lu.
\newblock Indifference pricing of a life insurance portfolio with risky asset driven by a shot-noise process.
\newblock \emph{Insurance: Mathematics and Economics}, 77:\penalty0 119--132, 2017.
\newblock ISSN 0167-6687.
\newblock \doi{https://doi.org/10.1016/j.insmatheco.2017.09.002}.

\bibitem[Luciano and Vigna(2008)]{LV2008}
E.~Luciano and E.~Vigna.
\newblock Mortality risk via affine stochastic intensities: Calibration and empirical relevance.
\newblock \emph{Belgian Actuarial Bulletin}, 8\penalty0 (1):\penalty0 5–16, 2008.
\newblock ISSN 1784-5742.

\bibitem[Ludkovski and Young(2008)]{LudkowskiYoung2008}
M.~Ludkovski and V.R. Young.
\newblock Indifference pricing of pure endowments and life annuities under stochastic hazard and interest rates.
\newblock \emph{Insurance: Mathematics and Economics}, 42(1):\penalty0 14--30, 2008.
\newblock \doi{10.1016/j.insmatheco.2006.11.009}.

\bibitem[Milevsky and Promislow(2001)]{MILPRO}
M.A. Milevsky and S.D. Promislow.
\newblock Mortality derivatives and the option to annuitise.
\newblock \emph{Insurance: Mathematics and Economics}, 29\penalty0 (3):\penalty0 299--318, 2001.
\newblock ISSN 0167-6687.
\newblock \doi{10.1016/S0167-6687(01)00093-2}.

\bibitem[Moore and Young(2003)]{YOUNGMOORE}
K.S. Moore and V.R. Young.
\newblock Pricing equity-linked pure endowments via the principle of equivalent utility.
\newblock \emph{Insurance: Mathematics and Economics}, 33\penalty0 (3):\penalty0 497--516, 2003.
\newblock \doi{10.1016/S0167-6687(03)00166-5}.

\bibitem[Pham(1998)]{pham1998optimal}
H.~Pham.
\newblock Optimal stopping of controlled jump diffusion processes: a viscosity solution approach.
\newblock In \emph{Journal of Mathematical Systems, Estimation and Control}. Citeseer, 1998.

\bibitem[Pirvu and Zhang(2013)]{pirvu2013utility}
T.~A. Pirvu and H.~Zhang.
\newblock Utility indifference pricing: a time consistent approach.
\newblock \emph{Applied Mathematical Finance}, 20\penalty0 (4):\penalty0 304--326, 2013.
\newblock \doi{10.1080/1350486X.2012.700575}.

\bibitem[Walter(1998)]{WALTER2}
W.~Walter.
\newblock \emph{Ordinary Differential Equations}, volume 182.
\newblock Springer Verlag, 1998.

\bibitem[Young(2003)]{Young}
V.R. Young.
\newblock Equity-indexed life insurance: Pricing and reserving using the principle of equivalent utility.
\newblock \emph{North American Actuarial Journal}, 7\penalty0 (1):\penalty0 68--86, 2003.
\newblock \doi{10.1080/10920277.2003.10596078}.

\bibitem[Young and Zariphopoulou(2002)]{YOUNGZAR}
V.R. Young and T.~Zariphopoulou.
\newblock Pricing dynamic insurance risks using the principle of equivalent utility.
\newblock \emph{Scandinavian Actuarial Journal}, 2002\penalty0 (4):\penalty0 246--279, 2002.
\newblock \doi{10.1080/03461230110106327}.

\end{thebibliography}

\appendix

\section{Technical proofs} \label{app:tech}

%\begin{proof}[Proof of Theorem \ref{thver}]

First, we provide the proof of the Verification Theorem for the pure investment problem.
\begin{proof}[Proof of Theorem \ref{thver} ]
The proof uses similar arguments as in those used to show Theorem \ref{thver2} below for the problem with the insurance derivative. Note that
Problem \ref{prbesp} corresponds to a special case of Problem \ref{prbcesp}, choosing $G_T = 0$. Nevertheless, for the sake of clarity we trace the fundamental steps of the proof.\\
By Proposition \ref{prop:value}, the function $\bar V(t, w, i)$ defined in equation \eqref{eq:valuefunction1} solves the HJB problem \eqref{VbarHJB}. Hence, for any $(t,w,i) \in [0,T] \times \R^+ \times \mathcal X$, we have
\begin{equation}\label{eqsolbar}
\bar {\mathcal L}_i^\Pi \bar{V}(s,W_{s}^\Pi(t,w),X_{s}(t,i)) \leq 0, \quad \forall s \in [t,T],\ \Pi \in \mathcal A_t, 
\end{equation}
where we recall that $\{W_{s}^{\Pi}(t,w),\ s \in [t,T]\}$ and $\{X_{s}(t,i),\ s \in [t,T]\}$ denote the solutions to equations \eqref{wealth} and \eqref{X} at time $s \in [t,T]$, starting from $(t,w) \in [0,T] \times \R^+$ and $(t,i) \in [0,T] \times \mathcal X$, respectively. Clearly, $\bar V(\cdot,\cdot,i) \in C^{1,2}([0,T]\times \R)$, for each $i \in \mathcal X$. 

In view of \eqref{wealth}, by applying It\^{o}'s formula, we have 
\begin{equation}\label{becb}
 \bar{V}(T,W_{T}^\Pi(t,w),X_{T}(t,i)) =  \bar{V}(t,\lambda,i) + \int_{t}^{T}\!\! \bar{\mathcal{L}}_i^{\Pi}\bar{V}(v,W_{v}^\Pi(t,w),X_{v\!}(t,i))\ud v + M_T,
\end{equation} 
where $M=\{M_r,\ r \in [t,T]\}$ is the stochastic process given by
\begin{equation}
\label{Mbar}
\begin{split}
M_r=& \int_{t}^{r}\Pi_{v}\sigma(v,X_{v})\pd{\bar{V}}{w}(v,W_{v}^\Pi,X_{v})  \ud Z_{v}^S \\ &  + \int_{t}^{r}\! \int_{\R} {{ \big{\{} \bar{V}\big{(}v,W_{v}^\Pi,X_{v-\!}+h(X_{v-\!},z) \big{)} - \bar{V}(v,W_{v}^\Pi,X_{v-\!}) \big{\}} \widehat{\mathcal{P}}(\ud v,\ud z)}} \\ & + \int_{t}^{r}\! {{ \big{\{} \bar{V}\big{(}v,W_{v-}^\Pi+ \Pi_v K_1(v,X_{v-}),X_{v-}) \big{)} - \bar{V}(v,W_{v-}^\Pi,X_{v-}) \big{\}} \{ \ud N^1_v - \Theta_1(v) \ud v \} }}  \\ &  + \int_{t}^{r}\! {{ \big{\{} \bar{V}\big{(}v,W_{v-}^\Pi- \Pi_v K_2(v,X_{v-}),X_{v-}) \big{)} - \bar{V}(v,W_{v-}^\Pi,X_{v-}) \big{\}} \{ \ud N^2_v - \Theta_2(v) \ud v \} }}.
\end{split}\end{equation}
In order to prove that $M$ is a $(\mathbb{G},\P)$-local martingale, %first we note that Problem \ref{prbesp} corresponds to a special case of Problem \ref{prbcesp}, choosing $G_T = 0$
we use a localization argument, %as in that of Theorem \ref{thver2} below for the problem with the insurance derivative, 
taking
\begin{equation}
\tau_{n}:= \inf \{s \in [t,T] \ | \ W_s^\Pi < -n \}, \quad n\in \N,
\end{equation}
which defines a non-decreasing sequence of stopping times $\{ \tau_n \}_{n \in \N}$ such that $\lim_{n\longrightarrow +\infty} \tau_n= +\infty$. 

Therefore, taking the conditional expectation with respect to $W_{t}=w$ and $X_{t}=i$ on both sides of \eqref{becb}, with $T$ replaced by $T \wedge \tau_n$, by \eqref{eqsolbar} we obtain that
\begin{equation}
\mathbb{E}_{t,w,i} \big{[} \bar{V}(T\wedge \tau_n,W_{T\wedge \tau_n}^\Pi(t,w),X_{T\wedge \tau_n}(t,i)) \big{]} 
%& \quad = \bar{V}(t,w,i) + \mathbb{E}_{t,w,i}\left[ \int_{t}^{T\wedge \tau_n}\left(\bar{{\mathcal L}}_i^\Pi V(r,W_{r}^\Pi(t,w),X_{r}(t,i))\right)\ud r\right] \\ 
\leq \bar{V}(t,w,i),	
\end{equation}
for every $\Pi \in \mathcal{A}_t$, $t \in \llbracket 0,T\wedge \tau_n\rrbracket$, $n \in \N$. Now, we note that 
\begin{equation}\begin{split}
&\mathbb{E} \big{[} \big(\bar{V}(T\wedge \tau_n,W_{T\wedge \tau_n}^\Pi(t,w),X_{T\wedge \tau_n}(t,i)) \big) ^2 \big{]} = \mathbb{E} \big{[} e^{-2\alpha W_{T\wedge \tau_n}^\Pi e^{r(T\wedge \tau_n -t)}} \varphi(T\wedge \tau_n,i)^2 \big{]}  < \infty. %\leq \tilde{K},	
\end{split}\end{equation} 
%for a positive constant $\tilde{K}$. 
Consequently, $ \{\bar{V}(T\wedge \tau_n,W_{T\wedge \tau_n}^\Pi(t,w),X_{T\wedge \tau_n}(t,i)) \}_{n\in \N}$ is a family of uniformly integrable random variables. Hence, it converges almost surely. Since $\{ \tau_n \}_{n\in \N}$ is a bounded and non-decreasing sequence of random times and the process $\{W_s^\Pi(t,w),\ s \in [t,T]\}$ is continuous in $T$, see \eqref{def:solwealth},
%and $\P(|W_t^\Pi|<+\infty)=1$, see \eqref{def:solwealth}, %we can apply the dominated convergence theorem and, 
taking the limit for $n \longrightarrow + \infty$, 
we get
\begin{equation} 
\begin{split}
&\mathbb{E}_{t,w,i} \big{[} \bar{V}(T,W_{T}^\Pi(t,w),X_{T}(t,i)) \big{]} = \lim_{n \longrightarrow +\infty} \mathbb{E}_{t,w,i} \big{[} \bar{V}(T\wedge \tau_n,W_{T\wedge \tau_n}^\Pi(t,w),X_{T\wedge \tau_n}(t,i)) \big{]} \\ 
& \quad \quad \leq \bar{V}(t,w,i), \quad \forall t \in [0,T],\ \Pi \in \mathcal{A}_t.
\end{split}
\end{equation}
%for every $\Pi \in \mathcal{A}_t$, $t \in [0,T]$.
%By the final condition in \eqref{VHJB} and the previous inequality, we get \begin{equation}
%\mathbb{E}_{t,w,i} \big{[} -e^{-\alpha (W_T^\Pi(t,w)-K)} \big{]} \leq \bar{V}(t,w,i),
%\end{equation} 
%for every $\Pi \in \mathcal{A}_t$, $t \in [0,T]$. 

As a byproduct, since $\Pi^*(t,i)$ given in Proposition \ref{opwi} realizes the infimum in \eqref{min-prob}, we have that $\bar {\mathcal L}_i^{\Pi^*}\bar V(t,w,i)=0$ and, performing the computations above, we get the equality
\begin{equation}
\mathbb{E} _{t,w,i} \Big{[} -e^{-\alpha W_{T}^{\Pi^*}(t,w)} \Big{]}=\sup_{\Pi \in \mathcal A_t} \mathbb{E} _{t,w,i} \Big{[} -e^{-\alpha W_{T}^{\Pi}(t,w)} \Big{]}=\bar{V} (t,w,i),
\end{equation}
that is, $\Pi_t^*=\Pi_t^*(t,X_t)$ is an optimal control.
\end{proof}

\begin{proof}[Proof of Lemma \ref{wysx}]
In view of \eqref{s}, \eqref{y} and  \eqref{wealth}, 
by applying It\^o's formula to the stochastic process $f (t,W_{t}^\Pi,\lambda_{t},X_{t})$, we have
\begin{equation}
f (t,W_{t}^\Pi,\lambda_{t},X_{t}) =  f(0,W_{0}^\Pi,\lambda_{0},X_{0}) + \int_0^t \mathcal L^\Pi f (u,W_u^\Pi,\lambda_u,X_u) \ud u + m_t,
\end{equation}
where
\begin{equation}\label{m}
\begin{split}
& m_t \\
& =  m_0 + \int_{0}^{t} \Pi_{v}\sigma(v,X_{v})\pd{f}{w}(v,W_{v}^\Pi,\lambda_{v},X_{v})\ud Z_{v}^S  + \int_{0}^{t}c(v,\lambda_{v})\lambda_v\pd{f}{\lambda}(v,W_{v}^\Pi,\lambda_{v},X_{v}) \ud Z_{v}^{\Lambda} \\ 
& + \int_{0}^{t}\! \int_{\R} {{ \Big( f\big{(}v,W_{v}^\Pi,\lambda_{v},X_{v-\!}+h(X_{v-\!},z) \big{)} - f(v,W_{v}^\Pi,\lambda_{v},X_{v-\!}) \Big) \widehat{\mathcal{P}}(\ud v,\ud z)}} \\ &  + \int_{0}^{t}\! {{ \Big( f\big{(}v,W_{v-}^\Pi+ \Pi_v K_1(v,X_{v-}),\lambda_{v},X_{v-}) \big{)} - f(v,W_{v-}^\Pi,\lambda_{v},X_{v-}) \Big)( \ud N^1_v - \Theta_1(v) \ud v ) }}  \\ &  + \int_{0}^{t}\! {{ \Big( f\big{(}v,W_{v-}^\Pi- \Pi_v K_2(v,X_{v-}),\lambda_{v},X_{v-}) \big) - f(v,W_{v-}^\Pi,\lambda_{v},X_{v-}) \Big) ( \ud N^2_v - \Theta_2(v) \ud v ) }}.	
\end{split} 
\end{equation}
We only need to prove that the process $m=\{m_t,\ t \in [0,T]\}$ is a $(\mathbb G, \P)$-martingale. 
%Since $f$ has bounded partial derivatives and by \eqref{int_ammiss} and \eqref{eq:exp_coeff}, respectively,
%recalling the definitions of $\sigma$ and $c$ and using \eqref{squadro} and \eqref{yquadro}, we have  
%\begin{equation*}
%\mathbb{E} \bigg[ \int_{0}^{T} \Big( \sigma(v,X_{v})\Pi_{v} \pd{f}{w}(v,W_{v}^\Pi,\lambda_{v},X_{v}) \Big)^{2} \ud v \bigg] < \infty;
%	\end{equation*} 
%\begin{equation*}
%\mathbb{E} \bigg[ \int_{0}^{T} \Big(  c(v,\lambda_{v})\pd{f}{\lambda}(v,W_{v}^\Pi,\lambda_{v},X_{v}) \Big)^{2} \ud v \bigg] < \infty.
%\end{equation*} 
%Thus, %according to the It\^{o} integral theory, 
By \eqref{ip1mkv}, the first two integrals in \eqref{m} are well-defined and turn out to be $(\mathbb G,\P)$-martingales.  
%	\begin{equation*}
%	\bigg( \int_{0}^{t}\sigma(v,X_{v-})\Pi_{v} \pd{f}{w} (v,W_{v},Y_{v},S_{v},X_{v-}) dZ_{v}\bigg)_{t\in [0,T]}
%	\end{equation*}\begin{equation*}
%	\bigg( \int_{0}^{t}\sigma(v,X_{v-})S_{v} \pd{f}{s} (v,W_{v},Y_{v},S_{v},X_{v-}) dZ_{v}\bigg)_{t\in [0,T]}
%	\end{equation*}
%	\begin{equation*}
%	\bigg( \int_{0}^{t}c(v,Y_{v}) \pd{f}{y} (v,W_{v},Y_{v},S_{v},X_{v-}) dZ^{Y}_{v}\bigg)_{t\in [0,T]}
%	\end{equation*} are well defined and they are $(\mathbb{G},\P)$-martingales.
Furthermore, due to \eqref{ip2mkv}, we have that also the jump terms in \eqref{m} are $(\mathbb{G},\P)$-martingales, (see e.g. \citet[Theorem $26.12(2)$]{MMO} and \citet[Lemma L$3$, Ch.II]{BREMAUD} for further details about the martingale property related to a Poisson random measure and a Poisson process, respectively). 	
\end{proof}
\noindent Now, we provide the proof of the Verification Theorem for the investment problem with the insurance derivative.
\begin{proof}[Proof of Theorem \ref{thver2}]
Let $\varphi :[0,T] \times \mathcal{X} \longrightarrow \R$ be a function 
%\aC{ho tolto bounded e che prende valori in $\R^+$, stessa cosa per $\phi$} 
such that $\varphi(\cdot,i) \in C^{1}\big{(} (0,T)\big{)} \cap  C\big{(} [0,T] \big{)}$, for each $i \in \mathcal{X}$, and suppose that it is a solution of the problem \eqref{PCvarphi}. Moreover, let $\phi :[0,T] \times \Rp \longrightarrow \R^+$ be a function such that $\phi(\cdot,\cdot) \in C^{1}\big{(} (0,T)\times \Rp \big{)} \cap  C\big{(} [0,T] \times \Rp \big{)}$, and suppose that it solves the problem \eqref{PCphi}. Now, taking $V$ defined in \eqref{V}, we have that
$V$ is a solution of the problem \eqref{VHJB}. This implies that, for every $(t,w,\lambda,i) \in [0,T]\times \R \times \Rp \times \mathcal{X}$
\begin{equation}\label{eq:vsol}
\begin{split}
&{\mathcal L}_i^\Pi V(s,W_s^\Pi(t,w),\lambda_s(t,\lambda),X_s(t,i)) \\
& + \lambda_s(t,\lambda) \left( \bar{V}(s,W_s^\Pi(t,w),X_s(t,i))- V(s,W_s^\Pi(t,w),\lambda_s(t,\lambda),X_s(t,i)) \right)\! \leq 0, \quad  s\! \in \![t,T],
\end{split}
\end{equation}
for all $\Pi \in \mathcal{A}_t$, where $\{\lambda_{s}(t,\lambda),\ s \in [t,T]\}$ denotes the solution to equation \eqref{y} with initial condition $\lambda_t= \lambda$ %$\{X_{r}(t,i),\ r \in [t,T]\}$ denotes the solution to equation \eqref{X} with initial condition $X_t=i$ and 
and $\bar V$ is the value function of the pure investment problem given in \eqref{Vbaruesp}.
In view of \eqref{wealth}, by applying It\^{o}'s formula, we have 
\begin{equation}\label{bec}
\begin{split}
& V(T,W_{T}^\Pi(t,w),\lambda_{T}(t,\lambda),X_{T}(t,i)\!) =  V(t,\lambda,i) + \!\!\int_{t}^{T}\!\!\!\! \mathcal{L}_i^{\Pi}V(v,W_{v}^\Pi(t,w),\lambda_{v}(t,\lambda),X_{v\!}(t,i))\ud v \\
& \  + \int_{t}^{T}\!\!\lambda_v(t,\lambda) \left( \bar{V}(v,W_v^\Pi(t,w),X_v(t,i))- V(v,W_v^\Pi(t,w),\lambda_v(t,\lambda),X_v(t,i)) \right) \ud v + M_T,
\end{split}
\end{equation} 
where $M=\{M_s,\ s \in [t,T]\}$ is the stochastic process given by
\begin{equation}
\label{M}
\begin{split}
&M_s= \int_{t}^{s}\Pi_{v}\sigma(v,X_{v})\pd{V}{w}(v,W_{v}^\Pi,\lambda_{v},X_{v})  \ud Z_{v}^S + \int_{t}^{s}c(v,Y_{v}) \lambda_v \pd{V}{\lambda}(v,W_{v}^\Pi,\lambda_{v},X_{v})  \ud Z_{v}^{\Lambda} \\ &  + \int_{t}^{s}\! \int_{\R} {{ \big{\{} V\big{(}v,W_{v}^\Pi,\lambda_{v},X_{v-\!}+h(X_{v-\!},z) \big{)} - V(v,W_{v}^\Pi,\lambda_{v},X_{v-\!}) \big{\}} \widehat{\mathcal{P}}(\ud v,\ud z)}} \\ &  + \int_{t}^{s}\! {{ \big{\{} V\big{(}v,W_{v-}^\Pi+ \Pi_v K_1(v,X_{v-}),\lambda_{v},X_{v-}) \big{)} - V(v,W_{v-}^\Pi,\lambda_{v},X_{v-}) \big{\}} \{ \ud N^1_v - \Theta_1(v) \ud v \} }}  \\ &  + \int_{t}^{s}\! {{ \big{\{} V\big{(}v,W_{v-}^\Pi- \Pi_v K_2(v,X_{v-}),\lambda_{v},X_{v-}) \big{)} - V(v,W_{v-}^\Pi,\lambda_{v},X_{v-}) \big{\}} \{ \ud N^2_v - \Theta_2(v) \ud v \} }}.
\end{split}\end{equation}
Now, we prove that $M$ is a $(\mathbb{G},\P)$-local martingale. Precisely, we need to show that 
\begin{equation*}
\mathbb{E} \bigg[ \int_{t}^{T\wedge \tau_n} \Big( \sigma(v,X_{v})\Pi_{v} \pd{V}{w}(v,W_{v}^\Pi,\lambda_{v},X_{v}) \Big)^{2} \ud v \bigg] < \infty,	\end{equation*}
\begin{equation*}
\mathbb{E} \bigg[ \int_{t}^{T\wedge \tau_n} \Big( c(v,\lambda_{v}) \lambda_v \pd{V}{\lambda}(v,W_{v}^\Pi,\lambda_{v},X_{v}) \Big)^{2} \ud v \bigg] < \infty,	\end{equation*}
for a suitable, non-decreasing sequence of stopping times $\{ \tau_n \}_{n \in \N}$ such that $\lim_{n\longrightarrow +\infty} \tau_n= +\infty$. 
Taking expression \eqref{V} into account, we note that 
\[\begin{array}{lll}
\pd{V}{w}(t,w,\lambda,i)&= \alpha \phi(t,\lambda) \varphi(t,i) e^{r(T-t)-\alpha w e^{r(T-t)}} %- \xi(t,y) e^{\eta(t,s)}\bar{V}(t,w,i)\alpha e^{r(i)(T-t)} 
, \\
\pd{V}{\lambda}(t,w,\lambda,i)&=-\pd{\phi}{\lambda}(t,\lambda)  \varphi(t,i) e^{-\alpha w e^{r(T-t)}} %\pd{\xi}{y}(t,y)  e^{\eta(t,s)}\bar{V}(t,w,i)
 %\pd{\eta}{s}(t,s)\xi(t,y)e^{\eta(t,s)}\bar{V}(t,w,i)
.\end{array}
\]
Let us define a sequence of random times  $\{ \tau_n \}_{n \in \N}$ by setting \begin{equation}
\tau_{n}:= \inf \{s \in [t,T] \ | \ W_s^\Pi < -n, \lambda_{s}>n, \phi(s,\lambda_s)>n, \pd{\phi}{\lambda}(s,\lambda_s)>n \}, \quad n \in \N.
\end{equation}
Throughout the proof, we denote by $C_n$ any constant depending on $n \in \N$.
Consequently,  we get 
\begin{equation*}
\begin{split}
& \mathbb{E} \bigg[ \int_0^{T\wedge \tau_n} \Big( \sigma(v,X_{v})\Pi_{v} \pd{V}{w}(v,W_{v}^\Pi,\lambda_{v},X_{v}) \Big)^{2} \ud v \bigg] \\ 
& \qquad \quad  = \mathbb{E} \bigg[ \int_0^{T\wedge \tau_n}  \sigma^{2}(v,X_{v})\Pi_{v}^{2} \Big(\alpha \phi (v,\lambda_{v}) \varphi(v,X_{v}) e^{r(T-v)-\alpha W^{\Pi}_{v} e^{r(T-v)}}  \Big)^{2} dv \bigg] \\ 
& \qquad \quad \leq C_{n} \mathbb{E} \bigg[ \int_0^{T} \sigma^2(v,X_v)\Pi_{v}  ^{2} dv \bigg] < \infty \ \ \ \forall n\in \N,
\end{split}
\end{equation*} 
since %$\varphi$ and $\phi$ are bounded by hypothesis and 
$\Pi$ is admissible. Further, %using conditions \eqref{paiii} and
by \eqref{eq:exp_coeff} we have that 
\begin{equation*}
\begin{split}
& \mathbb{E} \bigg[ \int_0^{T\wedge \tau_n} \Big( c(v,\lambda_{v}) \lambda_v \pd{V}{\lambda}(v,W_{v}^\Pi,\lambda_{v},X_{v}) \Big)^{2} \ud v \bigg]\\ 
& \qquad \quad  = \mathbb{E} \bigg[ \int_0^{T\wedge \tau_n}\bigg( c(v,\lambda_{v}) \lambda_v \pd{\phi}{\lambda}(v,\lambda_{v})  \varphi(v,X_{v}) e^{-\alpha W^{\Pi}_{v} e^{r(X_{v})(T-t)}} \bigg)^{2} \ud v \bigg] \\ 
%& \qquad \quad \leq \widetilde C \mathbb{E} \bigg[ \int_0^{T\wedge \tau_n}  \bigg( c(v,\lambda_{v})\lambda_v \varphi(v,X_{v}) e^{-\alpha W^{\Pi}_{v} e^{r(X_{v})(T-t)}} \bigg)^{2}\Big( 1+ |\lambda_{v}|^{\gamma}\Big)^{2} \ud v \bigg] \\ 
& \qquad \quad \leq C_{n} \mathbb{E} \bigg[ \int_0^{T}  c(v,\lambda_{v})^{2} \lambda_v^2 \ud v \bigg]< \infty \ \ \ \forall n \in \N.
\end{split}
\end{equation*} 
%since $\varphi$ and $\phi$ are bounded by hypothesis. % and using conditions \eqref{paiii} and \eqref{eq:exp_coeff}. 
Furthermore, due to the boundedness of function $V$ until time $\tau_n$, we have that
the stopped process \begin{equation}\left\{\int_{t}^{s\wedge\tau_n}\!\!\!\!\int_{\R} \!\!\! \Big{\{} V\big(v,W_{v}^\Pi,\lambda_{v},X_{v-}+h(X_{v-},z)\big) - V\big(v,W_{v}^\Pi,\lambda_{v},X_{v-}\big) \Big{\}}\mathcal{\widehat{P}}(\ud v,\ud z),\ s \!\in \![t,T]\!\right\} \end{equation} is a $(\mathbb{G},\P)$-martingale (see e.g. \citet[Theorem $26.12(2)$]{MMO}), for every $n \in \N$.
Finally, the stopped processes 
{\small
\begin{equation}
    \Big\{\!\int_{t}^{s\wedge\tau_n}\!\Big ( V\big(v,W_{v-}^\Pi + \Pi_v K_1(v,X_{v-}),\lambda_{v},X_{v-}\big) - V\big(v,W_{v-}^\Pi,\lambda_{v},X_{v}\big)\Big)( \ud N^1_v - \Theta_1(v) \ud v ) ,\ s \in [t,T]\Big\} 
\end{equation}}
and 
{\small
\begin{equation}\Big\{\!\int_{t}^{s\wedge\tau_n}\!\Big( V\big(v,W_{v-}^\Pi - \Pi_v K_2(v,X_{v-}),\lambda_{v},X_{v-}\big) - V\big(v,W_{v-}^\Pi,\lambda_{v},X_{v}\big) \Big)(\ud N^2_v - \Theta_2(v) \ud v) ,\ s \in [t,T]\Big\} \end{equation}}
\normalsize
are also $(\mathbb{G},\P)$-martingales, (see e.g. \citet[Lemma L$3$, Ch.II]{BREMAUD}). 
Thus, the process $\{ M_s, \ s \in [t,T] \}$ turns out to be a $(\mathbb{G},\P)$-local martingale and $\{ \tau_n \}_{n\in \N}$ is a localizing sequence for $\{ M_s, \ s \in [t,T] \}$.
Therefore, taking the conditional expectation of both sides of \eqref{bec} with respect to $W_{t}=w$, $\lambda_{t}=\lambda$ and $X_{t}=i$ with $T$ replaced by $T \wedge \tau_n$, by \eqref{eq:vsol} we obtain that
\begin{equation}
\begin{split}
&\mathbb{E}_{t,w,\lambda,i} \big{[} 	V(T\wedge \tau_n,W_{T\wedge \tau_n}^\Pi(t,w),\lambda_{T\wedge \tau_n},X_{T\wedge \tau_n}(t,i)) \big{]} 
%= V(t,w,\lambda,i)\\ 
%& \quad + \mathbb{E}_{t,w,\lambda,i}\left[ \int_{t}^{T\wedge \tau_n}\left({\mathcal L}_i^\Pi V(r,W_{r}^\Pi(t,w),\lambda_{r},X_{r}(t,i))\right)\ud r\right] \\
%& \qquad \qquad + \mathbb{E}_{t,w,\lambda,i}\left[  \int_{t}^{T\wedge \tau_n}\lambda_{r} \left( \bar{V}(r,W_{r}^\Pi(t,w),X_{r}(t,i))- V(r,W_{r}^\Pi(t,w),\lambda_{r},X_{r}(t,i)) \right) \ud r\right]\\
%&\quad 
\leq V(t,w,\lambda,i),	
\end{split}
\end{equation}
for every $\Pi \in \mathcal{A}_t$, $t \in \llbracket 0,T\wedge \tau_n\rrbracket$, $n\in \N$. Now, we note that 
\begin{equation}
\begin{split}
&\mathbb{E} \left[ \left(V(T\wedge \tau_n,W_{T\wedge \tau_n}^\Pi(t,w),\lambda_{T\wedge \tau_n}(t,\lambda),X_{T\wedge \tau_n}(t,i)) \right) ^2 \right] \\ &\quad = \mathbb{E} \left[ e^{-2\alpha W_{T\wedge \tau_n}^\Pi e^{r(T\wedge \tau_n -t)}} \varphi(T\wedge \tau_n,X_{T\wedge \tau_n})^2 \phi(T\wedge \tau_n,\lambda_{T\wedge \tau_n})^{2}\right]  \leq \widetilde{K},	
\end{split}
\end{equation} 
for a positive constant $\widetilde{K}$. This means that $ \{V(T\wedge \tau_n,W_{T\wedge \tau_n}^\Pi(t,w),\lambda_{T\wedge \tau_n},X_{T\wedge \tau_n}(t,i)) \}_{n \in \N}$ is a family of uniformly integrable random variables. Hence, it converges almost surely. Since $\{ \tau_n \}_{n \in \N}$ is a bounded and non-decreasing sequence of random times and the process $\{W_s^\Pi(t,w),\ s \in [t,T]\}$ is continuous in $T$, see \eqref{def:solwealth},
%$\P(|W_t^\Pi|<+\infty)=1$, see \eqref{def:solwealth}
in view of \eqref{cond:sup}, we can apply the dominated convergence theorem and, taking the limit for $n \longrightarrow + \infty$, we get
\begin{equation} 
\begin{split}
&\mathbb{E}_{t,w,\lambda,i} \big{[} V(T,W_{T}^\Pi(t,w),\lambda_T(t,\lambda),X_{T}(t,i)) \big{]}\\	& \quad \quad = \lim_{n \longrightarrow +\infty} \mathbb{E}_{t,w,\lambda,i} \big{[} V(T\wedge \tau_n,W_{T\wedge \tau_n}^\Pi(t,w),\lambda_{T\wedge \tau_n}(t,\lambda),X_{T\wedge \tau_n}(t,i)) \big{]} \leq V(t,w,\lambda,i),	
\end{split}
\end{equation}
for every $\Pi \in \mathcal{A}_t$, $t \in [0,T]$.
By the final condition in \eqref{VHJB} and the previous inequality, we get \begin{equation}
\mathbb{E}_{t,w,\lambda,i} \big{[} -e^{-\alpha (W_T^\Pi(t,w)-K)} \big{]} \leq V(t,w,\lambda,i),
\end{equation} 
for every $\Pi \in \mathcal{A}_t$, $t \in [0,T]$. 
Finally, since the insurance payment does not depend on the risky asset price, 
we have that $\Pi^{*}(t,w,\lambda,i)=\Pi^{*}(t,i)$ given in Proposition \ref{opwi} yields that
${\mathcal L}_i^{\Pi^*} V(t,w,\lambda,i) + \lambda\big( \bar{V}(t,w,i)- V(t,w,\lambda,i) \big) = 0$; then, if we apply the above arguments to $\Pi^*$
and replacing ${\mathcal L}_i^\Pi$ with ${\mathcal L}_i^{\Pi^*}$, 
we find the equality
\begin{equation}		
\sup_{\Pi \in \mathcal{A}_t} \mathbb{E}_{t,w,\lambda,i} \big{[} -e^{-\alpha (W_T^\Pi(t,w)-K)} \big{]} = V(t,w,\lambda,i),		\end{equation} 
which implies that the process $\Pi^*(t,X_t)$ is an optimal Markovian control.
\end{proof}

\section{Derivation of the HJB equation}\label{app:HJB}

For the sake of clarity, we show how to obtain a formal derivation of the HJB equation \eqref{VHJB} associated to the problem with the insurance derivative.
To this aim, we apply the Bellman's dynamic programming principle that, in this context, it is formulated as follows.
\begin{proposition}[Bellman optimality principle]
Let $(t,w,\lambda,i) \in [0,T] \times \R \times \Rp \times \mathcal{X}$. Then, for $t \leq t+h \leq T$ and $\Pi \in {\mathcal{A}_t}$, we have
%In every time interval $[t,t+h]$ of $[0,T]$, for $h>0$ such that $t+h<T$, the following property holds true: 
\begin{equation}\label{PPDV}
V(t,w,\lambda,i) \geq \mathbb{E}_{t,w,\lambda,i} \left[ V(t+h,W_{t+h}^\Pi,\lambda_{t+h},X_{t+h}) \right], 
\end{equation}
%for every $(w,y,s,i) \in \R \times \R \times \Rp \times \mathcal{X}$, 
where $V$ is the value function introduced in \eqref{Vuesp}.
Moreover, equality holds in \eqref{PPDV} if, and only if, the arbitrary control $\Pi$ on the interval $[t,t+h]$ is optimal.
\end{proposition}
\noindent The idea is that if the insurer follows the optimal strategy on $[t,T]$, her/his expected utility is at least as great as if she/he invests arbitrarily on $[t,t+h)$ and then optimally on $[t+h,T]$, for $h$ sufficiently small such that $t+h<T$. 
In the application of the dynamic programming principle, we must consider whether the policyholder survives from time $t$ until time $t+h$, as in \citet{YOUNGZAR}, \citet{YOUNGMOORE}, \citet{LudkowskiYoung2008} and \citet{Young}. 
Consider an individual aged $l$, who is seeking to buy a pure endowment policy. For the rest of this section, we write $(l)$ to refer to this individual. 
For each $h$ such that $t+h<T$, if the individual $(l+t)$ survives for another $h$ years until time $t+h$, which happens with probability $_{h}p_{l+t}$, the insurer still faces the endowment risk on the time interval $[t+h,T]$. In this case, by  \eqref{Vuesp}, the maximum expected utility derived by investing optimally on $[t+h,T]$ is $V(t+h,W_{t+h}^\Pi,\lambda_{t+h},X_{t+h})$. 
However, if the individual $(l+t)$ dies in $[t,t+h]$, an event that happens with probability $_{h}q_{l+t}$, then the insurer 
is not longer at risk for the endowment payout. Hence, by \eqref{Vbaruesp}, the maximum expected utility derived by investing optimally on $[t+h,T]$ is $\bar{V}(t+h,W_{t+h}^\Pi,X_{t+h})$.
From \eqref{PPDV}, we have 
{\small \begin{equation}
\begin{split}
V(t,w,\lambda,i)\geq \ _{h}p_{l+t} \mathbb{E}_{t,w,\lambda,i} \Big{[} V(t+h,W_{t+h}^\Pi,\lambda_{t+h},X_{t+h}) \Big{]}  + \ _{h}q_{l+t} \mathbb{E}_{t,w,i} \Big{[} \bar{V}(t+h,W_{t+h}^\Pi,X_{t+h}) \Big{]}. 
\end{split}
\end{equation}}
If we assume enough regularity conditions and appropriate integrability on the
value functions and their derivatives,
by applying It\^{o}'s formula and conditioning on $W_{t}^\Pi=w$, $\lambda_{t}=\lambda$ and $X_{t}=i$, we get
\begin{equation*}
\begin{aligned}
& V(t,w,\lambda,i) \geq  \ _{h}p_{l+t} V(t,w,\lambda,i) + _{h}q_{l+t}\bar{V}(t,w,i) \\ & \ + _{h}p_{l+t} \mathbb{E}_{t,w,\lambda,i} \bigg{[} \int_{t}^{t+h} \bigg{\{} \pd{V}{t} + \Big[ rW_{v}^{\Pi} + \big(\mu(v,X_{v})- r\big) \Pi_{v} \Big] \pd{V}{w} \ud v \bigg{\}} \bigg]  \\ & \ + _{h}p_{l+t} \mathbb{E}_{t,w,\lambda,i} \bigg[ \int_{t}^{t+h} \bigg{\{} b(v,\lambda_{v}) \lambda_v \pd{V}{\lambda} +  \frac{1}{2} \sigma^{2}(v,X_{v}) \Pi_{v}^{2} \pds{V}{w} +\frac{1}{2}c^{2}(v,\lambda_{v}) \lambda_v^2 \pds{V}{\lambda}\bigg{\}}  \ud v \bigg{]} \\ & \ + _{h}p_{l+t} \mathbb{E}_{t,w,\lambda,i} \bigg{[} \int_{t}^{t+h} \bigg{\{}\sum_{j \in \mathcal{X}} V(v,W_{v}^{\Pi},\lambda_{v},j)a_{v,j} \bigg{\}}  \ud v \bigg{]} \\ & \ + _{h}p_{l+t} \mathbb{E}_{t,w,\lambda,i} \bigg{[} \int_{t}^{t+h} \Theta_1(v) \big{\{} V(v,W_v^{\Pi}+\Pi_v K_1(v,i),\lambda_v,X_v) -V(v,W^{\Pi}_v,\lambda_v,X_v) \big{\}} \ud v \bigg{]} \\ & \ + _{h}p_{l+t} \mathbb{E}_{t,w,\lambda,i} \bigg{[} \int_{t}^{t+h} \Theta_2(v) \big{\{} V(v,W_v^{\Pi}-\Pi_v K_2(v,i),\lambda_v,i) -V(v,W_v^{\Pi},\lambda_v,X_v) \big{\}}  \ud v \bigg{]}
\\ & \ + _{h}q_{l+t} \mathbb{E}_{t,w,i} \bigg{[} \int_{t}^{t+h} \bigg{\{} \pd{\bar{V}}{t} + \Big[ rW_{v}^{\Pi} + \big(\mu(v,X_{v})- r \big)\Pi_{v} \big] \pd{\bar{V}}{w} \bigg{\}}  \ud v \bigg{]} \\ &  \ + _{h}q_{l+t} \mathbb{E}_{t,w,i} \bigg{[} \int_{t}^{t+h} \bigg{\{} \frac{1}{2} \sigma(v,X_{v})^{2} \Pi_{v}^{2} \pds{V}{w} + \sum_{j \in \mathcal{X}} \bar{V}(v,W_{v},j)a_{v,j} \bigg{\}} \ud v \bigg{]} \\ & \ + _{h}q_{l+t} \mathbb{E}_{t,w,i} \bigg{[} \int_{t}^{t+h} \Theta_1(v) \big{\{} \bar{V}(v,W_v^{\Pi}+\Pi_v K_1(v,i),X_v) -\bar{V}(v,W^{\Pi}_v,X_v) \big{\}} \ud v \bigg{]} \\ & \ + _{h}q_{l+t} \mathbb{E}_{t,w,i} \bigg{[} \int_{t}^{t+h} \Theta_2(v) \big{\{} \bar{V}(v,W_v^{\Pi}-\Pi_v K_2(v,i),i) -\bar{V}(v,W_v^{\Pi},X_v) \big{\}}  \ud v \bigg{]}.
\end{aligned}
\end{equation*}
To keep the formulas readable, in the integrals above we have suppressed the independent variables $(v,W_{v},\lambda_{v},X_v)$ and $(v,W_{v},X_v)$ of the partial derivatives of $V$ and $\bar V$, respectively.
By subtracting $_{h}p_{l+t}V(t,w,\lambda,i)$ from both sides of inequality and dividing both sides by $h$, we obtain 
{\small \begin{equation*}
\begin{aligned}
& \frac{_{h}q_{l+t}}{h}V(t,w,\lambda,i) \geq  \ \frac{_{h}q_{l+t}}{h} \bar{V}(t,w,i) \\ & \ + _{h}p_{l+t} \mathbb{E}_{t,w,\lambda,i} \bigg{[} \int_{t}^{t+h} \frac{1}{h}\bigg{\{} \pd{V}{t} + \Big[ rW_{v}^{\Pi} + \big(\mu(v,X_{v})- r\big)\Pi_{v} \Big] \pd{V}{w} \ud v \bigg{\}} \bigg]  \\ & \ + _{h}p_{l+t} \mathbb{E}_{t,w,\lambda,i} \bigg[ \int_{t}^{t+h}\frac{1}{h} \bigg{\{} b(v,\lambda_{v}) \lambda_v \pd{V}{\lambda} +   \frac{1}{2} \sigma^{2}(v,X_{v}) \Pi_{v}^{2} \pds{V}{w} +\frac{1}{2}c^{2}(v,\lambda_{v\!})\lambda_v^2 \pds{V}{\lambda} \bigg{\}}  \ud v \bigg{]} \\ & \ + _{h}p_{l+t} \mathbb{E}_{t,w,\lambda,i} \bigg{[} \int_{t}^{t+h}\frac{1}{h} \bigg{\{}\sum_{j \in \mathcal{X}} V(v,W_{v},\lambda_{v},j)a_{v,j} \bigg{\}}  \ud v \bigg{]} \\ & \ + _{h}p_{l+t} \mathbb{E}_{t,w,\lambda,i} \bigg{[} \int_{t}^{t+h} \frac{1}{h} \bigg{\{} \Theta_1(v) \big{\{} V(v,W_v^{\Pi}+\Pi_v K_1(v,i),\lambda_v,X_v) -V(v,W^{\Pi}_v,\lambda_v,X_v) \big{\}} \bigg{\}} \ud v \bigg{]} \\ & \ + _{h}p_{l+t} \mathbb{E}_{t,w,\lambda,i} \bigg{[} \int_{t}^{t+h} \frac{1}{h} \bigg{\{} \Theta_2(v) \big{\{} V(v,W_v^{\Pi}-\Pi_v K_2(v,i),\lambda_v,i) -V(v,W_v^{\Pi},\lambda_v,X_v) \big{\}} \bigg{\}} \ud v \bigg{]}
\\ & \ + _{h}q_{l+t} \mathbb{E}_{t,w,i} \bigg{[} \int_{t}^{t+h}\frac{1}{h} \bigg{\{} \pd{\bar{V}}{t} + \Big[ rW_{v} + \big(\mu(v,X_{v})- r \big)\Pi_{v} \big] \pd{\bar{V}}{w} \bigg{\}}  \ud v \bigg{]} \\ & \ + _{h}q_{l+t} \mathbb{E}_{t,w,i} \bigg{[} \int_{t}^{t+h} \frac{1}{h}\bigg{\{} \frac{1}{2} \sigma(v,X_{v})^{2} \Pi_{v}^{2} \pds{V}{w} + \sum_{j \in \mathcal{X}} \bar{V}(v,W_{v}^{\Pi},j)a_{v,j} \bigg{\}} \ud v \bigg{]} \\ & \ + _{h}q_{l+t} \mathbb{E}_{t,w,i} \bigg{[} \int_{t}^{t+h} \frac{1}{h} \bigg{\{}\Theta_1(v) \big{\{} \bar{V}(v,W_v^{\Pi}+\Pi_v K_1(v,i),X_v) -\bar{V}(v,W^{\Pi}_v,X_v) \big{\}} \bigg{\}} \ud v \bigg{]} \\ & \ + _{h}q_{l+t} \mathbb{E}_{t,w,i} \bigg{[} \int_{t}^{t+h} \frac{1}{h} \bigg{\{} \Theta_2(v) \big{\{} \bar{V}(v,W_v^{\Pi}-\Pi_v K_2(v,i),i) -\bar{V}(v,W_v^{\Pi},X_v) \big{\}} \bigg{\}} \ud v \bigg{]}. \end{aligned}
\end{equation*}}
We observe that as $h\longrightarrow 0^{+}$, we have 
{\small \begin{equation}
_{h}p_{l+t} \longrightarrow 1 , \ \ \ _{h}q_{l+t} \longrightarrow 0  \  \ \ \mbox{and} \ \ \ \frac{_{h}q_{l+t}}{h} \longrightarrow \lambda_{t},
\end{equation}} for each $t \in [0,T]$. Consequently, taking the limit as $h\longrightarrow 0^{+}$ yields 
{\small \begin{equation*}
\begin{split}
0 & \geq  \lambda \big( \bar{V}(t,w,i)-V(t,w,\lambda,i) \big) + \pd{V}{t} + \big[ rw + \big(\mu(t,i) - r\big)\Pi \big] \pd{V}{w} + b(t,\lambda)\lambda \pd{V}{\lambda} \\ &  + \frac{1}{2}\Pi^{2}\sigma^{2}(t,i) \pds{V}{w}+ \frac{1}{2}c^{2}(t,\lambda)\lambda^2\pds{V}{\lambda} +
\sum_{j \in \mathcal{X}} V(t,w,\lambda,j)a_{ij}  \\& + \Theta_1(t) \big{\{} V(t,w+\Pi K_1(t,i),\lambda,i) -V(t,w,\lambda,i) \big{\}} \\
& + \Theta_2(t) \big{\{} V(t,w-\Pi K_2(t,i),\lambda,i) -V(t,w,\lambda,i) \big{\}}.
\end{split}
\end{equation*} }
Finally, we note that along the optimum, we have
{\small \begin{equation*}
\begin{split}
& 0 =   \lambda \left( \bar{V}(t,w,i)- V(t,w,\lambda,i) \right) + \pd{V}{t} + rw\pd{V}{w} +  b(t,\lambda)\lambda \pd{V}{\lambda} + \frac{1}{2} c^{2}(t,\lambda)\lambda^2\pds{V}{\lambda} + \sum_{j \in \mathcal{X}} V(t,w,\lambda,j)a_{ij} \\ &  + \sup_{\Pi \in \R} \Bigg{[} \big(\mu(t,i)- r\big)\Pi \pd{V}{w}  + \frac{1}{2} \sigma^{2}(t,i) \Pi^{2} \pds{V}{w}  + \Theta_1(t) \big{\{} V(t,w+\Pi K_1(t,i),\lambda,i) -V(t,w,\lambda,i) \big{\}} \\ & + \Theta_2(t) \big{\{} V(t,w-\Pi K_2(t,i),\lambda,i) -V(t,w,\lambda,i) \big{\}} \Bigg{]}, \quad \forall (t,w,\lambda,i) \in [0,T) \times \R \times \Rp \times \mathcal X,
\end{split}
\end{equation*}}
with $V(T,w,\lambda,i)=-e^{-\alpha(w-K)}$, for each $(w,\lambda,i) \in \R \times \Rp \times \mathcal X$,
which coincides with \eqref{VHJB}.

\end{document}